\documentclass[acmlarge,nonacm]{acmart}

\AtBeginDocument{%
  \providecommand\BibTeX{{%
    \normalfont B\kern-0.5em{\scshape i\kern-0.25em b}\kern-0.8em\TeX}}}

\setcopyright{acmcopyright}
\copyrightyear{2018}
\acmYear{2018}
\acmDOI{XXXXXXX.XXXXXXX}

\acmJournal{JACM}
\acmNumber{2}
\acmMonth{4}

\usepackage{amsmath,amsfonts,stmaryrd,amssymb}
\usepackage{hyperref}
\usepackage{todonotes}
\usepackage{xcolor}
\hypersetup{
	colorlinks,
	linkcolor={red!50!black},
	citecolor={blue!50!black},
	urlcolor={blue!80!black}
}

\usepackage{wrapfig}
\usepackage{graphicx}
\usepackage{subcaption}
\usepackage{cleveref} 
\usepackage{quiver} 
\usepackage{proof} 
\usepackage{scalerel} 

\usepackage{tikzit}

\tikzstyle{observation}=[text=blue]
\tikzstyle{morphism}=[fill=white, draw=black, shape=rectangle]
\tikzstyle{small box}=[fill=white, draw=black, shape=rectangle, minimum width=0.4cm, minimum height=0.45cm]
\tikzstyle{medium box}=[fill=white, draw=black, shape=rectangle, minimum width=0.8cm, minimum height=0.9cm]
\tikzstyle{large morphism}=[fill=white, draw=black, shape=rectangle, minimum width=1.7cm, minimum height=1cm]
\tikzstyle{bn}=[fill=black, draw=black, shape=circle, inner sep=1.5pt]
\tikzstyle{state}=[fill=white, draw=black, regular polygon, regular polygon sides=3, minimum width=0.8cm, shape border rotate=180, inner sep=0pt]
\tikzstyle{effect}=[fill=white, draw=black, regular polygon, regular polygon sides=3, minimum width=0.8cm, shape border rotate=0, inner sep=0pt]
\tikzstyle{medium state}=[fill=white, draw=black, regular polygon, regular polygon sides=3, minimum width=1.3cm, inner sep=0pt, shape border rotate=180]
\tikzstyle{small state}=[fill=white, draw=black, regular polygon, regular polygon sides=3, minimum width=0.7cm, inner sep=0pt, shape border rotate=180]
\tikzstyle{medium effect}=[fill=white, draw=black, regular polygon, regular polygon sides=3, minimum width=1.3cm, inner sep=0pt, shape border rotate=0]
\tikzstyle{large state}=[fill=white, draw=black, regular polygon, regular polygon sides=3, minimum width=2.2cm, shape border rotate=180, inner sep=0pt]
\tikzstyle{wn}=[fill=white, draw=black, shape=circle, inner sep=1.5pt]
\tikzstyle{treenode}=[fill=white, draw=none, shape=circle]

\tikzstyle{arrow}=[->]
\tikzstyle{dashed box}=[-, dashed]
\tikzstyle{condition}=[draw=blue, dashed]

\pgfdeclarelayer{edgelayer}
\pgfdeclarelayer{nodelayer}
\pgfsetlayers{edgelayer,nodelayer,main}
\tikzstyle{none}=[]

\tikzstyle{observation}=[text=blue]
\tikzstyle{morphism}=[fill=white, draw=black, shape=rectangle]
\tikzstyle{small box}=[fill=white, draw=black, shape=rectangle, minimum width=0.4cm, minimum height=0.45cm]
\tikzstyle{medium box}=[fill=white, draw=black, shape=rectangle, minimum width=0.8cm, minimum height=0.9cm]
\tikzstyle{large morphism}=[fill=white, draw=black, shape=rectangle, minimum width=1.7cm, minimum height=1cm]
\tikzstyle{bn}=[fill=black, draw=black, shape=circle, inner sep=1.5pt]
\tikzstyle{state}=[fill=white, draw=black, regular polygon, regular polygon sides=3, minimum width=0.8cm, shape border rotate=180, inner sep=0pt]
\tikzstyle{effect}=[fill=white, draw=black, regular polygon, regular polygon sides=3, minimum width=0.8cm, shape border rotate=0, inner sep=0pt]
\tikzstyle{medium state}=[fill=white, draw=black, regular polygon, regular polygon sides=3, minimum width=1.3cm, inner sep=0pt, shape border rotate=180]
\tikzstyle{small state}=[fill=white, draw=black, regular polygon, regular polygon sides=3, minimum width=0.7cm, inner sep=0pt, shape border rotate=180]
\tikzstyle{medium effect}=[fill=white, draw=black, regular polygon, regular polygon sides=3, minimum width=1.3cm, inner sep=0pt, shape border rotate=0]
\tikzstyle{large state}=[fill=white, draw=black, regular polygon, regular polygon sides=3, minimum width=2.2cm, shape border rotate=180, inner sep=0pt]
\tikzstyle{wn}=[fill=white, draw=black, shape=circle, inner sep=1.5pt]
\tikzstyle{treenode}=[fill=white, draw=none, shape=circle]

\tikzstyle{arrow}=[->]
\tikzstyle{dashed box}=[-, dashed]
\tikzstyle{condition}=[draw=blue, dashed]

\tikzset{baseline=(current  bounding  box.center)}

\usepackage{color}
\definecolor{deepblue}{rgb}{0,0,0.5}
\definecolor{deepred}{rgb}{0.6,0,0}
\definecolor{deepgreen}{rgb}{0,0.5,0}
\definecolor{darkgray}{rgb}{0.5,0.5,0.5}

\usepackage{listings}

\DeclareFixedFont{\ttb}{T1}{txtt}{bx}{n}{9} 
\DeclareFixedFont{\ttm}{T1}{txtt}{m}{n}{9}  

\lstdefinestyle{python_ppl}{
	language=Python,
	basicstyle=\ttm,
	otherkeywords={let, to},          
	keywordstyle=\ttb\color{deepblue},
	emph={Gauss,N,condition, =:=, observe, sample, score, normal, gp_sample},     
	emphstyle=\ttb\color{deepred},    
	stringstyle=\color{deepgreen}  
}

\lstdefinelanguage{CustomML}{
	keywords={match, with, rec, true, false, fun, return, let, in, if, then, else, type, val, module, sig, end, ref, struct},
	keywordstyle=\color{deepblue}\bfseries,
	ndkeywords={ref},
	ndkeywordstyle=\color{darkgray}\bfseries,
	identifierstyle=\color{black},
	sensitive=false,
	comment=[l]{//},
	morecomment=[s]{/*}{*/},
	commentstyle=\color{darkgray}\ttfamily,
	stringstyle=\color{red}\ttfamily,
	morestring=[b]',
	morestring=[b]"
}

\lstdefinestyle{ml_ppl}{
	language=CustomML,
	basicstyle=\ttm,
	otherkeywords={},          
	keywordstyle=\ttb\color{deepblue},
	emph={Gauss,N,condition, =:=, observe, sample, score, normal, gp_sample},     
	emphstyle=\ttb\color{deepred},    
	stringstyle=\color{deepgreen}
}

\lstdefinelanguage{JavaScript}{
	keywords={typeof, new, true, false, catch, function, return, null, catch, switch, var, if, in, while, do, else, case, break},
	keywordstyle=\color{blue}\bfseries,
	ndkeywords={class, export, boolean, throw, implements, import, this},
	ndkeywordstyle=\color{darkgray}\bfseries,
	identifierstyle=\color{black},
	sensitive=false,
	comment=[l]{//},
	morecomment=[s]{/*}{*/},
	commentstyle=\color{darkgray}\ttfamily,
	stringstyle=\color{red}\ttfamily,
	morestring=[b]',
	morestring=[b]"
}

\lstdefinestyle{webppl}{
	language=JavaScript,
	basicstyle=\ttm,
	otherkeywords={let},            
	keywordstyle=\ttb\color{deepblue},
	emph={flip, condition, factor, sample, normal, score, observe},
	emphstyle=\ttb\color{deepred},    
	stringstyle=\color{deepgreen}
}

\lstdefinestyle{prolog}{
	language=Prolog,
	basicstyle=\ttm,         
	keywordstyle=\ttb\color{deepblue},
	emphstyle=\ttb\color{deepred},    
	stringstyle=\color{deepgreen}, 
}

\lstdefinestyle{funprolog}{
	language=Prolog,
	basicstyle=\ttm,
	otherkeywords={let, in, return},            
	keywordstyle=\ttb\color{deepblue},
	emphstyle=\ttb\color{deepred},    
	stringstyle=\color{deepgreen}, 
}

\newcommand*{\mlstinline}[1]{\text{\lstinline|#1|}}
\newcommand*{\code}[1]{\lstinline|#1|}

\providecommand{\C}{}
\providecommand{\G}{}
\newcommand{\fv}{\mathsf{fv}} 
\newcommand{\cat}{\mathbb} 
\renewcommand{\C}{\cat C} 
\renewcommand{\G}{\mathcal G} 
\renewcommand{\d}{\mathrm d} 

\newcommand{\catname}{\mathsf} 
\newcommand{\set}{\catname{Set}}

\newcommand{\sbs}{\catname{Sbs}}

\newcommand{\finstoch}{\catname{FinStoch}}
\newcommand{\finprojstoch}{\catname{FinProjStoch}}
\newcommand{\finsubstoch}{\catname{FinSubStoch}}
\newcommand{\borelstoch}{\catname{BorelStoch}}

\newcommand{\gauss}{\catname{Gauss}}

\newcommand{\rel}{\catname{Rel}}

\newcommand{\del}{\mathrm{del}}
\newcommand{\cpy}{\mathrm{copy}}

\newcommand{\kl}{\catname{Kl}}

\newcommand{\kerto}{\leadsto} 

\newcommand{\N}{\mathcal{N}}

\newcommand{\id}{\mathsf{id}}

\newcommand{\swap}{\mathsf{swap}}
\newcommand{\ev}{\mathrm{ev}}
\newcommand{\supp}{\mathsf{supp}}

\newcommand{\R}{\mathbb R}

\newcommand{\s}{\,|\,}

\newcommand{\sem}[1]{{\llbracket #1 \rrbracket}}

\newcommand{\letin}[2]{\mathsf{let}\,#1\,=\,#2\,\mathsf{in}\,}
\newcommand{\return}{\mathrm{return}\,}

\newcommand{\true}{\mathrm{true}}
\newcommand{\false}{\mathrm{false}}

\newcommand{\ite}[3]{\mathsf{if}~#1~\mathsf{then}~#2~\mathsf{else}~#3}

\renewcommand{\det}{\mathrm{det}}
\newcommand{\cond}{\catname{Cond}}
\newcommand{\obs}{\catname{Obs}}

\newcommand{\col}{\mathrm{col}}

\newcommand{\type}[1]{\mathsf{#1}}

\newcommand{\tunit}{\type{unit}}

\newcommand{\Id}{\mathrm{Id}}

\newcommand{\obsto}{\leadsto}
\newcommand{\red}{\vartriangleright}
\newcommand{\rv}{\mathsf{R}}
\newcommand{\unit}{\mathsf{I}}

\newcommand{\normal}{\mathsf{normal}}
\newcommand{\bernoulli}{\mathsf{bernoulli}}

\newcommand{\eq}{\mathrel{\scalebox{0.4}[1]{${=}$}{:}{\scalebox{0.4}[1]{${=}$}}}}
\newcommand{\eqo}{{:}\scalebox{0.5}[1]{${=}$}\,}

\newcommand{\defeq}{\stackrel{\text{def}}=}

\newcommand{\ppi}{\mathbf{P}}
\newcommand{\ppisl}{\mathbf{PSL}}

\newcommand{\norm}{\mathsf{normalize}}

\newcommand{\subd}{D_{\leq 1}}

\newcommand{\ct}[1]{\underline{#1}}

\begin{document}

\title{Probabilistic Programming with Exact Conditions}

\author{Dario Stein}
\email{dario.stein@ru.nl}
\affiliation{%
  \institution{Radboud University Nijmegen}
  \streetaddress{Erasmusplein 1}
  \city{Nijmegen}
  \country{The Netherlands}
}

\author{Sam Staton}
\affiliation{%
  \institution{University of Oxford}
  \streetaddress{Parks Road, OX1 3QD}
  \city{Oxford}
  \country{United Kingdom}}
\email{sam.staton@cs.ox.ac.uk}

\renewcommand{\shortauthors}{Stein and Staton}

\begin{abstract}
We spell out the paradigm of \emph{exact conditioning} as an intuitive and powerful way of conditioning on observations in probabilistic programs. This is contrasted with likelihood-based \emph{scoring} known from languages such as \textsc{Stan}. We study exact conditioning in the cases of discrete and Gaussian probability, presenting prototypical languages for each case and giving semantics to them. We make use of categorical probability (namely Markov and CD categories) to give a general account of exact conditioning which avoids limits and measure theory, instead focusing on restructuring dataflow and program equations. The correspondence between such categories and a class of programming languages is made precise by defining the internal language of a CD category.
\end{abstract}

\ccsdesc[500]{Theory of computation~Denotational semantics}
\ccsdesc[500]{Theory of computation~Categorical semantics}
\ccsdesc[500]{Theory of computation~Probabilistic computation}


\maketitle

\section{Introduction}\label{sec:intro} \lstset{style=python_ppl}

Probabilistic programming is a programming paradigm that uses code to formulate generative statistical models and perform inference on them~\cite{van_de_meent:ppl_introduction, probmods2}. Techniques from programming language theory can be used to understand modelling assumptions such as conditional independence, as well as enable optimizations that improve the efficiency of various inference algorithms (e.g. \cite{staton:commutative_semantics}). We'll also elaborate on this application in \Cref{sec:extendedintro}.

There are two different styles of conditioning on data in a probabilistic program: \emph{scoring} and \emph{exact conditioning}. \emph{Scoring} features constructs to re-weight the current execution trace of the probabilistic program with a given likelihood. By contrast, \emph{exact conditioning} focuses on a primitive operation $E_1 \eq E_2$ which signifies that expressions $E_1$ and $E_2$ shall be conditioned to be exactly equal. A prototypical exact conditioning program looks as follows, where we infer some underlying value \lstinline|x| from a noisy measurement \lstinline|y| (see also \Cref{sec:extendedintro}):

\begin{lstlisting}
x = normal($\mu$=50, $\sigma$=10) # prior
y = normal($\mu$=x, $\sigma$=5)     # noisy measurement
y =:= 40               # make exact observation
\end{lstlisting}

Variants of exact conditioning are available in different frameworks: In \textsc{Hakaru} \citep{shan_ramsey}, certain exact conditioning queries can be addressed using symbolic disintegration, but $(\eq)$ is not a first-class construct in Hakaru. \textsc{Infer.NET} \citep{InferNET18} does allow exact conditioning on variables and employs an approximate inference algorithm to solve the resulting queries (e.g. \citep{infernet_tutorial}). The intended formal meaning of exact conditioning is however far from obvious when continuous distributions such as Gaussians are involved (in our example, the observation \lstinline|y == 40| has probability zero). The goal of this article is to rigorously spell out the exact conditioning paradigm, give semantics to it and analyze its properties.

We note that, even in this simple example, the use of exact conditioning is intuitive and allows the programmer to cleanly decouple the generative model from the data observation stage.
As the example makes clear, exact conditioning lends itself to logical reasoning about programs. For example, after conditioning $s \eq t$, the expressions $s$ and $t$ are known to be equal and can be interchanged.

As we will show, exact conditioning also enjoys good formal properties which allow us to simplify programs compositionally.
Among the desired properties are the following: Program lines can be reordered as long as dataflow is respected. That is, the \emph{commutativity equation} remains valid for programs with conditioning
	\begin{equation}
	\begin{array}{l}
	\letin x u \\
	\letin y v t
	\end{array}
	\equiv
	\begin{array}{l}
	\letin y v \\
	\letin x u t
	\end{array} \label{eqn:initcommutativity}
	\end{equation}
	where $x$ not free in $v$ and $y$ not free in $u$. We have a \emph{substitution law}: if $t\eq u$ appears in a program, then later occurrences of $t$ may be replaced by $u$.
	\begin{equation}(t\eq u);v[t/x]\quad \equiv
	\quad (t\eq u);v[u/x]
	\label{eqn:substlong}
	\end{equation}
 As a special base case, if we condition a normal variable on a constant $\underline c$, then the variable is simply \emph{initialized} to this value
	\begin{equation}\letin x {\normal()} {(x\eq \underline c);t} \quad \equiv
	\quad t[\underline c/x]
	\label{eqn:initlong}
	\end{equation}
We substantiate these claims further and give examples of applications of these laws in our extended introduction (\Cref{sec:extendedintro}). \\

\noindent In order to formally study exact conditioning, we focus on two concrete fragments of probabilistic computation
\begin{enumerate}
\item finite probability, which deals with finite sets and discrete distributions
\item Gaussian probability, which deals with multivariate Gaussian (normal) distributions and affine-linear maps
\end{enumerate}
We give probabilistic languages for each fragment and extend them with an exact conditioning construct. The language for finite probability is well-known, while the Gaussian language is novel. We give its formal description and operational semantics in \Cref{sec:gaussian_language}. \\

\textsc{Methods:} Our goal is to give denotational semantics to these languages, and prove that the desired properties from \cref{eqn:initcommutativity,eqn:substlong,eqn:initlong} hold. We wish to do this in a abstract way, that relies as little as possible on the particular details of finite or Gaussian probability, but instead treats them uniformly and is open to generalization. \emph{Categorical probability theory} is such an abstract language of mathematical models of probability, which allows us to discuss the relevant notions such as determinism, independence and conditioning in a uniform way. We connect categorical probability theory to probabilistic programming using the following Curry-Howard style correspondence
\begin{enumerate}
\item probabilistic programs are the \emph{internal languages} of categorical probability theories; program terms $t$ can be interpreted as morphisms $\sem{t}$ in these categories
\item for every probabilistic language, its \emph{syntactic category} is a categorical model of probability theory. objects are types and morphisms are terms of the language
\end{enumerate}
We prove a particular version of the correspondence in \Cref{sec:cd-calc}: The CD-calculus is the internal language of CD categories, a widespread model of categorical probability theory \Cref{sec:cd_cats}. This language, which resembles a first-order OCaml, serves as a meta-language of which all other languages discussed in this article (finite or Gaussian, with or without conditioning) will be particular instances. 

The Curry-Howard correspondence gives us \emph{three equivalent formalisms} for describing probabilistic models (see \Cref{fig:composition})
\begin{enumerate}
\item terms in a probabilistic language
\item morphisms in a categorical model
\item \emph{string diagrams}, which are a well-known graphical notation to describe compositions of morphisms in an intuitive way \cite{selinger:graphical}
\end{enumerate}  
We will frequently convert back and forth between the different formalisms for convenience, conciseness and to emphasize different mathematical or programming intuitions. We expect some familiarity of the reader for translating string-diagrammatic and algebraic categorical notation, though we will give a brief reminder in the introduction of \Cref{sec:synth_conditioning}. The correspondence between string diagrams and program terms is a novel technical contribution of this article. It is formally proved in \Cref{sec:cd-calc}, and we will aid the reader by spelling out examples in both programming and categorical terms. A reader not primarily interested in category theory may still view our string diagram manipulations as a concise graphical way of encoding program transformations. 

\begin{figure}[h]
	\centering
	\begin{subfigure}[b]{0.3\textwidth}
		\[ \input{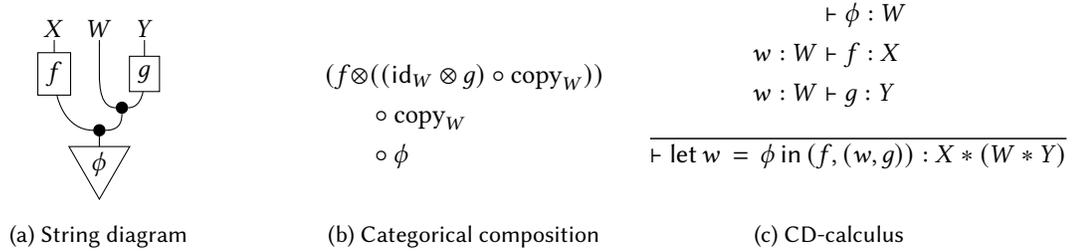} \]
		\caption{String diagram}
		\label{fig:stringdiag}
	\end{subfigure}
	\begin{subfigure}[b]{0.3\textwidth}
		\begin{align*}
			(f \otimes &((\id_W \otimes g) \circ \cpy_W)) \\ & \circ \cpy_W \\ &\circ \phi 
		\end{align*}
		\vspace{0.0cm}
		\caption{Categorical composition}
		\label{fig:categorical}
	\end{subfigure} 
	\begin{subfigure}[b]{0.3\textwidth}
		\begin{align*}
			&\vdash \phi : W \\
			w : W &\vdash f : X \\
			w : W &\vdash g : Y 
		\end{align*}
		\begin{align*}
			\overline{\vdash \letin w \phi (f,(w,g)) : X \ast (W \ast Y)}
		\end{align*} 
		\vspace{0.0cm}
		\caption{CD-calculus}
		\label{fig:cdcalc}
	\end{subfigure}
	\caption{Different formalisms for composition}\label{fig:composition}
\end{figure}

\noindent To connect our development to more traditional probabilistic notions, we consider the formalism of directed graphical models (Bayesian networks). For example, the three expressions of \Cref{fig:composition} can all be understood as encoding the generative structure of the following Bayesian network
\begin{equation} \begin{tikzpicture}[scale=0.6]
	\begin{pgfonlayer}{nodelayer}
		\node [draw=black, shape=circle] (0) at (-2, 0) {$X$};
		\node [draw=black, shape=circle] (1) at (2, 0) {$Y$};
		\node [draw=black, shape=circle] (2) at (0, -1.75) {$W$};
	\end{pgfonlayer}
	\begin{pgfonlayer}{edgelayer}
		\draw [style=arrow] (2) to (0);
		\draw [style=arrow] (2) to (1);
	\end{pgfonlayer}
\end{tikzpicture}
 \label{eq:graphicalmodel} \end{equation}
This Bayesian network indicates that there are random variables valued in spaces $W$, $X$ and $Y$, with the given conditional independence structure, but this conditional independence means that we can describe the situation by giving the distribution $\phi$ of the random variable valued in $W$, together with the distributions of the random variables valued in $X$ and in $Y$, which both depend on the random choice for $W$, via $f$ and~$g$.

A systematic comparison of graphical models and string diagrams is given in \cite{fong}. We note that Bayesian networks are a strictly weaker formalism than the other three, i.e. not every string diagram or probabilistic program can be obtained from a Bayesian network.

\paragraph{Contributions: }
We build up this categorical machinery to first understand the semantics of probabilistic languages \emph{without conditioning}. We then use the same framework to extend the language with an exact conditioning operator. On the side of categorical semantics, this amounts to extending a suitable category $\C$ (for conditioning-free computation) with effects $X \to I$ for conditioning on observations. We call the extended category $\cond(\C)$. Constructing this category (\Cref{sec:cond}) and proving its desirable properties is the central contribution of this article. 

Importantly, the $\cond$ construction makes no mention of measure theory, densities or limits, which usually feature in discussions of conditional probability, but is defined purely in terms of the categorical structure of $\C$. It is closely tied to reasoning in terms of program transformations: The Cond construction can be understood as giving normal forms for straight-line programs with exact observations, modulo contextual equivalence
\[ x : X \vdash \letin {(y,k) : Y \otimes K} {h} {(k \eqo o); \return y \quad : Y} \]
Our definition of \emph{abstract inference problems} in \Cref{sec:abstract_cond} follows that intuition. The desired properties of exact conditioning can be proved purely abstractly for the $\cond$ construction. Our type-theoretic approach to conditioning will also help addressing counterintuitive behavior such as Borel's paradox (\Cref{sec:outlook_paradoxes}). \\

We give concrete descriptions of the $\cond$ construction applied to our main examples of discrete and Gaussian probability. This means analyzing contextual equivalence for our example languages in detail: For finite probability, we obtain substochastic kernels modulo \emph{automatic renormalization} (\Cref{sec:finprojstoch}). This fully characterizes contextual equivalence for straight-line inference with discrete probability, and refines the semantics using the subdistribution monad. Our discussion reveals interesting connections between the admissibility of automatic normalization and the expressibility of branching in probabilistic programs (\Cref{sec:straightline}). For the Gaussian language we give a concrete analysis by proving that the denotational semantics is fully abstract, in \Cref{sec:denotational}. 

\paragraph{Outline: }

We give an overview over the structure of the article: 

\begin{itemize}
\item To demonstrate the strengths and intricacies of the exact conditioning approach, we will follow up this introduction with an extended example (\Cref{sec:extendedintro}) elaborating the noisy measurement example, and demonstrate the power of program transformations and compositional reasoning for a Gaussian random walk example. 
	
\item In \Cref{sec:gaussian_language}, we formally introduce the Gaussian language and its operational semantics.

\item In \Cref{sec:synth_conditioning}, we review string diagrams and categorical probability theory using CD- and Markov categories. We then introduce the CD-calculus (\Cref{sec:cd_cats}) and prove it to be internal language of CD categories, which gives us the central correspondence of probabilistic programs and string diagrams. In \Cref{sec:abstract_cond}, we use the categorical notions to develop an abstract theory of inference problems in Markov categories.

\item In \Cref{sec:cond}, we present the $\cond$ construction, which is the centerpiece of this article. We prove the well-definedness of its construction (\Cref{thm:functoriality}) and verify the desired laws for conditioning in full generality (\Cref{sec:laws_cond}).

\item In \Cref{sec:finstoch}, we return to analyze the $\cond$ construction in detail for our two example settings. This is tantamount to studying contextual equivalence for our exact conditioning languages: In \Cref{sec:denotational}, we show that the denotational semantics for the Gaussian language is fully abstract. In \Cref{sec:finprojstoch}, we conduct a similar analysis for finite probability, arriving at an explicit characterization of $\cond(\mathsf{FinStoch})$. Our discussion reveals interesting connections between the admissibility of automatic normalization and the availability of branching in probabilistic programs (\Cref{sec:straightline}).
\end{itemize}

\noindent \emph{Note.} The starting point for this article is our paper in the Proceedings of LICS 2021~\cite{2021compositional}, where we introduced the Gaussian language, and used the $\cond$ construction to prove a full abstraction result. For this invited journal submission, we expand upon \cite{2021compositional} with greater detail and background throughout Sections~\ref{sec:intro}-\ref{sec:abstract_cond}, and expose the $\cond$ construction in a self-contained manner with an emphasis on program equations and graphical reasoning in \Cref{sec:cond}. This submission also incorporates otherwise unpublished material that is part of the first author's DPhil thesis~\cite{dariothesis}, such as the formal presentation of the CD calculus. The treatment of finite probability in \Cref{sec:finstoch} and the recognition of the special role of branching in \Cref{sec:straightline} are entirely novel contributions of this article. A Python implementation of the Gaussian Language is available under \cite{gaussianinfer}. 

\subsection{Extended Discussion about Exact Conditioning}\label{sec:extendedintro}

We proceed with an extended discussion on the differences between the scoring and exact conditioning paradigms, and the strengths and difficulties related to exact conditioning. This discussion uses an informal Python-like language, and is not technically essential for the rest of the paper. The later sections of the article are fully formal. \pagebreak

\textbf{Exact Conditioning versus Scoring: } In the introduction, we considered an noisy measurement example: Our \emph{prior} assumption about the distribution of some quantity $X$ is that it is normally distributed with mean $\mu=50$ and standard deviation $\sigma=10$. We only have access to a noisy measurement $Y$, which itself has standard deviation $5$, and observe a value of $Y=40$. Conditioned on that observation, the \emph{posterior} distribution over $X$ is now $\N(42,\sigma)$ with $\sigma = \sqrt{20} \approx 4.47$. 

In probabilistic programming with scoring, the primitive \lstinline|score(r)| re-weights the current execution trace of the probabilistic program with a \emph{score} or \emph{likelihood} $\mlstinline{r} \in \mathbb R_{> 0}$. A derived operation \lstinline|observe(x,D)| expresses an observation of a value $x$ from some distribution $D$ by scoring with the density $r = \mathsf{pdf}_D(x)$. The scoring implementation of the noisy measurement example therefore looks like this:
\begin{lstlisting}
x = normal(50, 10)  # prior
observe(40, normal(x,5))  # observation
\end{lstlisting}
The idea of Monte Carlo simulation is to run the program many times, picking different values for \lstinline|x| from the normal distribution, but preferring runs with a high likelihood. This makes execution traces more likely whose value of \mlstinline{x} lies closer to $40$.
Scoring constructs are widely available in popular probabilistic languages such as \textsc{Stan} \citep{stan} or \textsc{WebPPL} \citep{dippl}. Scoring with likelihoods from $\{0,1\}$ is sometimes called a \emph{hard constraint}, as opposed to more general \emph{soft constraints}. The prototypical way of performing inference on scoring programs is by likelihood-weighted importance sampling. Hard constraints turn this into mere rejection sampling, because likelihood-zero traces are discarded entirely. Replacing hard constraints by equivalent soft ones can thus be beneficial for inference efficiency. 

Exact conditions are strictly more powerful than scoring, because we can express \lstinline|observe(x,D)| in terms of conditioning on a freshly generated sample as \lstinline|let y = sample(D) in y =:= x|. On the other hand, not every exact conditioning program can be expressed in terms of scoring:

In the special case of discrete probability, we can express an exact condition $E_1 \eq E_2$ by the hard constraint \lstinline|score(if $E_1$ == $E_2$ then 1 else 0)| without issue. This causes an execution trace to be discarded whenever the condition is not met. This encoding is no longer viable for continuous distributions such as Gaussians: For example, the program
\begin{lstlisting}
x = normal(0,1); x =:= 40
\end{lstlisting}
should return \lstinline|x=40| deterministically, because $x$ is conditioned to have that value. On the other hand, the following hard constraint
\begin{lstlisting}
x = normal(0,1); score(if x == 40 then 1 else 0)
\end{lstlisting}
will \emph{reject every execution trace}, because the probability that \lstinline|x==40| is true equals zero. It is important to distinguish exact conditioning $(\eq)$ from the boolean equality test $(==)$. This distinction is crucial to making sense of apparent paradoxes such as Borel's paradox (\Cref{sec:outlook_paradoxes}). \\

\textbf{Compositional Reasoning about Conditions: } To elaborate the power of reasoning compositionally about conditioning programs, we consider the example of a simple Gaussian random walk with 100 steps, together with a table \lstinline|obs| of exact observations (\Cref{fig:rw}). A straightforward implementation would be to first generate the entire random walk, and then condition on the observations
\begin{lstlisting}
# generative model
for i in range(1,100):
  y[i] = y[i-1] + normal(0,1)
# observations
for j in obs:
  y[j] =:= obs[j]
\end{lstlisting}

The same program is more complicated to express without exact conditioning: Using soft conditions, the observations would need to be known at the time of generation and \lstinline|observe| commands need to be issued in-place, breaking the decoupling between the model and the data. 

On the other hand, rewriting the original model in such a way may improve the efficiency of inference. We can verify such a transformation using compositional reasoning: As we will show, it is consistent to reorder program lines as long as the dataflow is respected \eqref{eqn:initcommutativity}, so the random walk program is equivalent to the following version with interleaved observations
\begin{lstlisting}
for i in range(1,100):
  y[i] = y[i-1] + normal(0,1)
  if i in obs: y[i] =:= obs[i]
\end{lstlisting}
In the observation branch, we can now use initialization principle \eqref{eqn:initlong} to set \lstinline|y[i]| to its target value directly as \lstinline|y[i] = obs[i]|. The remaining condition becomes \lstinline|(y[i] - y[i-1]) =:= normal(0,1)| so we obtain 
\begin{lstlisting}
for i in range(1,100):
  if i in obs:
    y[i] = obs[i]
    (y[i] - y[i-1]) =:= normal(0,1) 
  else:
    y[i] = y[i-1] + normal(0,1)
\end{lstlisting}
In this version of the program, all exact conditions can now be replaced by \lstinline|observe| statements:
\begin{lstlisting}
for i in range(1,100):
  if i in obs:
    y[i] = obs[i]
    observe(y[i] - y[i-1] , normal(0,1))
  else:
    y[i] = y[i-1] + normal(0,1)
\end{lstlisting}
We can run this resulting program directly using a Monte Carlo simulation in Stan or WebPPL. \\

\begin{figure}[h]
	\centering
	\includegraphics[width=0.40\textwidth]{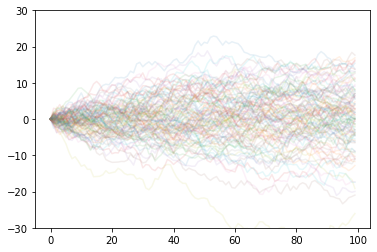}
	\includegraphics[width=0.40\textwidth]{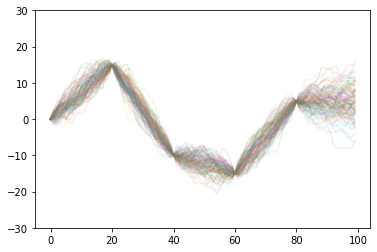}
	\caption{Gaussian random walk (left) and conditioned posterior (right) with four exact observations at $i=20,40,60,80$\label{fig:rw}}
\end{figure}

In \Cref{sec:gaussian_language}, we formalize the operational semantics of our Gaussian language, in which there are two key commands: drawing from a standard normal distribution ($\normal()$) and exact conditioning $(\eq)$. The operational semantics is defined in terms of configurations $(t,\psi)$ where $t$ is a program and $\psi$ is a state, which here is a Gaussian distribution. Each call to $\normal()$ introduces a new dimension into the state~$\psi$, and conditioning $(\eq)$ alters the state~$\psi$, using a canonical form of conditioning for Gaussian distributions (\Cref{sec:recap_gauss}).

In our first version of the random walk example, the operational semantics will first build up the prior distribution shown on the left in \Cref{fig:rw}, and then the second part of the program will condition to yield a distribution as shown on the right. But for the other programs above, the conditioning will be interleaved in the building of the model.

In stateful programming languages, composition of programs is often complicated and local transformations are difficult to reason about. This is what makes programs transformations like the ones we used powerful and nontrivial to verify.


\section{A Language for Gaussian Probability}
\label{sec:gaussian_language}
In this section, after an overview of the mathematics of Gaussian probability (\Cref{sec:recap_gauss}), we formally introduce a typed language (\Cref{sec:types}) for Gaussian probability and exact conditioning, and provide an operational semantics for it (\Cref{sec:opsem}). 
Our operational semantics is straightforward, in that it maintains a symbolic description of the distribution over all latent variables as the program runs, expressed as a covariance matrix. Thus the aim of this section is to formally describe language that we are studying in this paper. 

Looking beyond this section, the aim of the further sections of this paper is to address that issue that, as usual, the simple operational semantics here is intensional and non-compositional. It is intensional in that if two different programs actually behave in the same way, that might be very unclear from the operational semantics; it is non-compositional in that the role of running subprograms is hidden in the overall run of the operational semantics. The aim of the remainder of the paper, then, is to establish an equational and denotational framework for exact conditioning.  

The language in this section is focused on Gaussian probability, for concreteness, but to understand the equational framework it will be helpful in future sections to move to a general setting (\Cref{sec:cd-calc} and \Cref{sec:cond} for the general case without and with exact conditioning respectively). Once this general framework is established, we are able to offer a denotational explanation of exact conditioning, which specializes to this Gaussian language (\Cref{sec:denotational}). 

\subsection{Recap of Gaussian Probability}
\label{sec:recap_gauss}
We briefly recall \emph{Gaussian probability}, by which we mean the treatment of multivariate Gaussian distributions and affine-linear maps (e.g. \citep{lauritzen}). A \emph{Gaussian distribution} is the law of a random vector $X \in \R^n$ of the form $X=AZ + \mu$ where $A \in \R^{n \times m}$, $\mu \in \R^n$ are not random but the vector $Z$ is a \emph{multivariate standard normal} random vector. That is, its components $Z_1, \ldots, Z_m \sim \mathcal N(0,1)$ are independent and standard normally distributed, i.e.~with the following probability density function:
\[ \varphi(x) = \frac{1}{\sqrt{2\pi}} e^{-\frac 1 2 x^2} \]
(Formally, this is a density is regarded with respect to the Lebesgue measure, see Section~\ref{app:meas}, but a high-school knowledge of probability density is sufficient for this section.)
The distribution of $X$ is fully characterized by its \emph{mean} $\mu$ and the positive semidefinite \emph{covariance matrix} $\Sigma=AA^T$. Conversely, for any $\mu$ and positive semidefinite matrix $\Sigma$ there is a unique Gaussian distribution of that mean and covariance denoted $\mathcal N(\mu, \Sigma)$. The vector $X$ takes values precisely in the affine subspace $S=\mu + \col(\Sigma)$ where $\col(\Sigma)$ denotes the column space of $\Sigma$. We call $S$ the \emph{support} of the distribution. We note that while it is common to consider Gaussian distributions with nonsingular (positive definite) covariance matrix, it is convenient to allow the more general positive semidefinite case here, even including vanishing covariance. Natural programming constructs such as the copying of variables results in a singular covariance, and setting $\Sigma = 0$ lets us treat deterministic computation as a special case of probabilistic one. On the flipside, singular covariance requires us to carefully consider supports in our theory of conditioning. \\

Gaussian probability defines a small convenient fragment of probability theory, with the following properties:
\begin{itemize}
\item Affine transformations of Gaussians remain Gaussian. That is, if an affine map $f$ is written as $f(x)=Ax+b$ and $X$ is a random vector with mean $\mu$ and covariance $\Sigma$, then $f(X)$ has mean $A\mu + b$ and covariance $A\Sigma A^T$. This operation is called the \emph{pushforward} of distributions and is denoted $f_*$, i.e.
\begin{equation} f_*(\mathcal N(\mu,\Sigma)) = \mathcal N(A\mu + b, A\Sigma A^T) \label{def:pushforward} \end{equation}
\item Conditional distributions of Gaussians are themselves Gaussian. If we decompose an $(m+n)$-dimensional Gaussian vector $X \sim \mathcal N(\mu, \Sigma)$ into components $X_1, X_2$ with
\begin{equation*}
X = \begin{pmatrix} X_1 \\ X_2 \end{pmatrix}, \mu = \begin{pmatrix} \mu_1 \\ \mu_2 \end{pmatrix}, \Sigma = \begin{pmatrix} 
\Sigma_{11} & \Sigma_{12} \\ \Sigma_{21} & \Sigma_{22}
\end{pmatrix}\text{ where } \Sigma_{21} = \Sigma_{12}^T
\end{equation*}
there is a well-known explicit formula (e.g. \citep[3.13]{eaton}) for the conditional distribution $X_1|(X_2 = a)$ of $X_1$ conditional on $X_2=a$ for $a \in \supp(X_2)$. Namely $X_1|(X_2 = a) \sim \mathcal N(\mu',\Sigma')$ where
\begin{equation}
\mu' = \mu_1 + \Sigma_{12}\Sigma_{22}^{-}(a-\mu_2) \quad \Sigma' = \Sigma_{11} - \Sigma_{12}\Sigma_{22}^-\Sigma_{21}
\label{eq:conjugacy_formula}
\end{equation}
and $\Sigma_{22}^-$ is any \emph{generalized inverse} of $\Sigma_{22}$ .
\end{itemize}
We elaborate on formula \eqref{eq:conjugacy_formula} a bit: A generalized inverse of an $(m \times n)$-matrix $M$ is an $(n \times m)$-matrix $M^-$ such that $MM^-M = M$. Such inverses can be shown to always exist, but they need not be unique. If $M$ is invertible, its unique generalized inverse is $M^{-1}$. The posterior covariance matrix $\Sigma' =  \Sigma_{11} - \Sigma_{12}\Sigma_{22}^-\Sigma_{21}$ is also known as \emph{Schur complement} and is independent of the choice of generalized inverse. The matrix $\Sigma_{12}\Sigma_{22}^{-}$ appearing in the calculation of $\mu'$ \emph{does} depend on the choice of $\Sigma_{22}^-$. However, it takes uniquely defined values on the subspace $\col(\Sigma_{22})$. Therefore, formula \eqref{eq:conjugacy_formula} is only well-defined if the observation $a$ lies in the support of $X_2$. This caveat is mirrored in our categorical treatment of \Cref{sec:abstract_cond}, where conditionals are only unique on supports. A popular choice of generalized inverse is the Moore-Penrose pseudoinverse, which has connections to least squares optimization. For a detailed discussion of these concepts, we refer to \citep[Section~1.6]{zhang_schur}. 

The formula for conditional probability becomes particular simple if we condition on a single real-valued component of a vector: Let $X \sim \mathcal N(\mu, \Sigma)$ and let $Z = uX$ for some $u \in \R^{n \times 1}$, then the covariance of $(X,Z)$, regarded as a random $(n+1)$-vector, decomposes as \[\begin{pmatrix} 
    \Sigma & \Sigma u^T \\ u\Sigma^T & \sigma_{22}
  \end{pmatrix}
\qquad \text{where $\sigma_{22} = u\Sigma u^T$}\]
and the conditional distribution of $X|(Z=a)$ is $\N(\mu',\Sigma')$ with
\begin{equation}
\mu' = \mu + \frac{a-u\mu}{\sigma_{22}}\Sigma u^T, \quad \Sigma' = \Sigma - \frac 1 {\sigma_{22}} \Sigma u^Tu \Sigma \label{eq:conjugacy_simple}
\end{equation}
whenever $\sigma_{22} > 0$. If $\sigma_{22} = 0$ and $u\mu = a$, the condition is tautologously $0=0$ and we have $\mu' = \mu$, $\Sigma'=\Sigma$. Otherwise, $a \notin \supp(Z)$, and the conditioning problem has no well-defined solution.

\begin{example}
	\label{ex:og_gaussian_diagonal}
	Let $X,Y \sim \mathcal N(0,1)$ be independent and $Z=X-Y$. The joint distribution of $(X,Y,Z)$ is $\N(\vec 0,\Sigma)$ with covariance matrix
   \[ \Sigma = \begin{pmatrix} 1 & 0 & 1 \\ 0 & 1 & -1 \\ 1 & -1 & 2 \end{pmatrix} \]
	By \eqref{eq:conjugacy_simple}, the conditional distribution of $(X,Y)$ given $Z=0$ has the following covariance matrix
	\[ \Sigma' = \begin{pmatrix} 1 & 0  \\ 0 & 1 \end{pmatrix} - \frac 1 2 \begin{pmatrix} 1 \\ -1 \end{pmatrix} \cdot \begin{pmatrix} 1 & -1 \end{pmatrix} = \begin{pmatrix} 0.5 & 0.5 \\ 0.5 & 0.5 \end{pmatrix} \]
	The posterior distribution is thus equivalent to the model
	\[ X \sim \mathcal N(0,0.5), Y=X \]
        with one univariate Normal distribution having mean $0$ and variance $0.5$. 
      \end{example}

\paragraph{Borel's paradox}
Borel's paradox is an important subtlety that occurs when conditioning on the equality of random variables $X=Y$. The original formulation involves conditioning a uniform point on a sphere to lie on a great circle, but we will use Borel's paradox to refer to any situation where conditioning on equivalent equations leads to different outcomes (e.g. \citep{shan_ramsey}). For example, if instead of the condition $X-Y=0$ in  \Cref{ex:og_gaussian_diagonal} we had chosen the seemingly equivalent equations $X/Y=1$ or even $[X=Y]=1$ (using Iverson bracket notation), we would have obtained different posteriors:

\begin{example}\label{ex:borels_paradox}
If $X, Y \sim \mathcal N(0,1)$, then conditioned on $(X/Y=1)$, the variable $X$ can be shown to have density $|x|e^{-x^2}$ \citep{expect_the_unexpected}. Under the boolean condition $[X=Y]=1$, the inference problem has no solution because the model $X, Y \sim \mathcal N(0,1), Z = [X=Y]$ is measure-theoretically equal to $X, Y \sim \mathcal N(0,1), Z = 0$ (since independent Gaussian random variables are almost surely different), and conditioning on $0=1$ is inconsistent.
\end{example}

We will address Borel's paradox and posit that a careful type-theoretic phrasing (\Cref{sec:abstract_cond}) helps alleviate its seemingly paradoxical nature (\Cref{sec:outlook_paradoxes}). 

\subsection{Types and Terms of the Gaussian language}
\label{sec:types}
We now describe a language for Gaussian probability and conditioning. The core language resembles first-order OCaml with a construct $\normal()$ to sample from a standard Gaussian, and conditioning denoted as $(\eq)$. Types $\tau$ are generated from a basic type $\rv$ denoting \emph{real number} or \emph{random variable}, pair types and unit type $I$. 
\[ \tau ::= \rv \s \unit \s \tau \ast \tau \]
Terms of the language are
\begin{align*}
e ::= x &\s e + e \s \alpha \cdot e \s \underline{\beta} \s (e,e) \s () \\
&\s \letin x e e \s \letin {(x,y)} e e \\
&\s \normal() \s e \eq e 
\end{align*}
where $\alpha,\beta$ range over real numbers. Typing judgements are
\[ \infer{\Gamma, x : \tau, \Gamma' \vdash x : \tau}{}  
\qquad
\infer{\Gamma \vdash () : \unit}{}
\qquad
\infer{\Gamma \vdash (s,t) : \sigma \ast \tau}{\Gamma \vdash s : \sigma \quad \Gamma \vdash t : \tau}  
\]
\[ \infer{\Gamma \vdash s + t : \rv}{\Gamma \vdash s : \rv \quad \Gamma \vdash t : \rv} 
\qquad
\infer{\Gamma \vdash \alpha \cdot t : \rv}{\Gamma \vdash t : \rv}
\qquad
\infer{\Gamma \vdash \underline{\beta} : \rv}{}
\]
\[ \infer{\Gamma \vdash \normal() : \rv}{} 
\qquad
\infer{\Gamma \vdash (s \eq t) : \unit}{\Gamma \vdash s : \rv \quad \Gamma \vdash t : \rv}
\]
\[ \infer{\Gamma \vdash \letin x s t : \tau}{\Gamma \vdash s : \sigma \quad \Gamma, x : \sigma \vdash t : \tau}
\]
\[ 
\infer{\Gamma \vdash \letin {(x,y)} s t : \tau}{\Gamma \vdash s : \sigma \ast \sigma' \quad \Gamma, x : \sigma, y : \sigma' \vdash t : \tau}
\]

In Section~\ref{sec:cd-calc} we will introduce the general CD-calculus, and our Gaussian language is an instance of this\footnote{in the CD-calculus, we use projection maps rather than pattern-matching $\mathsf{let}$, but those constructs are interdefinable}, with base type $\rv$ and signature
\begin{align}
 (+) : \rv \ast \rv \to \rv, \quad \alpha \cdot (-) : \rv \to \rv, \quad \underline \beta : \unit \to \rv, \quad \normal : \unit \to \rv, \quad (\eq) : \rv \ast \rv \to \unit \label{eq:sig_gauss}
\end{align}
This will give us a clear path to denotational semantics: In \Cref{sec:denotational}, we will indeed identify our language as the internal language of an appropriate CD category with an exact conditioning morphism. 

We use standard syntactic sugar for sequencing $s;t$. We identify the type $\rv^n = \rv \ast (\rv \ast \ldots )$ with vectors $x=(x_1,\ldots,x_n$), and write matrix-vector multiplication $A \cdot x$ in an informal manner. For $\sigma \in \R$ and $e : \rv$, we define $\normal(x,\sigma^2) \equiv x + \sigma \cdot \normal()$. More generally, for a covariance matrix $\Sigma$ and $x : \rv^n$, we write $\normal(x, \Sigma) = x + A\cdot (\normal(), \ldots, \normal())$ where $A$ is any matrix such that $\Sigma = AA^T$. By the simple nature of the typing rules, we can identify any context and type with $\rv^n$ for suitable $n$.

For example, referring to \Cref{ex:og_gaussian_diagonal}, the tuple $(X,Y,Z)$ can be written in our language as
\[\letin {(x,y)} {(\normal(),\normal())} {(x,y,x-y)}
\]
The full example with conditioning can be written
\[\letin {(x,y,z)}{(\letin {(x,y)} {(\normal(),\normal())} {(x,y,x-y)})} {z\eq 0;(x,y)}
\]
This program is contextually equivalent (\Cref{def:ctxequiv}) to
\[\letin {x} {\sqrt{0.5} * \normal()} {\letin y x {(x,y)}}
\]

\subsection{Operational Semantics}
\label{sec:opsem}

Informally, our operational semantics works as follows: calling $\normal()$ allocates a latent random variable, and a prior distribution over all latent variables is maintained; calling $(\eq)$ updates this prior by symbolic inference according to the formula \eqref{eq:conjugacy_formula}.

Formally, we define a reduction relation over configurations.
A \emph{configuration} is either a dedicated failure symbol $\bot$ or a pair \[(e,\psi)\] where $\psi$ is a Gaussian distribution on $\R^r$ (i.e.~a mean vector and covariance matrix) and $z_1 : \rv, \ldots, z_r : \rv \vdash e$.
Thus a running term~$e$ may have free variables; these stand for dimensions in a given multivariate Gaussian distribution~$\psi$, reminiscent of a closure in a higher-order language. 

To define a reduction relation, we first introduce values, redexes and reduction contexts. 
\emph{Values} $v,w$ and \emph{redexes}~$\rho$ are defined as 
\begin{align*}
v,w &::= x \s (v,w) \s v + w \s \alpha \cdot v \s \underline{\beta} \s () \\
\rho &::= \normal() \s v \eq w \s \letin x v e \s \letin {(x,y)} v e
\end{align*}
A \emph{reduction context} $C$ with hole $[-]$ is of the form
\begin{align*}
C ::= [-] &\s (C,e) \s (v,C) \s C + e \s v + C \s \alpha \cdot C \s C \eq e \s v \eq C \\
&\s \letin x C e \s \letin {(x,y)} C e 
\end{align*}
Perhaps the only thing to note is that, in keeping with the call-by-value tradition of most probabilistic programming languages, we do reduce before a let assignment, i.e.~$\letin x C e$ is a reduction context.
It is easy to show by induction that every term is either a value or decomposes uniquely as $C[\rho]$. The latent variables $(z_1\dots z_r)$ are taken from a distinct supply of variable names $\{z_i : i \in \mathbb N \}$. We first define reduction on redexes (1--3), and then reduction contexts (4):
\begin{enumerate}
	\item Calling $\normal()$ allocates a fresh latent variable and adds an independent dimension to the prior
          \[ (\normal(),\psi) \red (z_{\mathrm{r+1}}, \psi \otimes \mathcal N(0,1))
          \]
          where $\psi$ is a Gaussian distribution on $\R^r$ with mean $\mu$ and covariance $\Sigma$;
          here $(\psi \otimes \mathcal N(0,1))$ is notation for the Gaussian distribution on $\R^{r+1}$
          with mean $(\mu,0)$ and covariance matrix $(\begin{smallmatrix}\Sigma & 0 \\ 0 & 1\end{smallmatrix})$. 
	\item To define conditioning, note that every value \mbox{$z_1 : \rv , \ldots, z_r : \rv \vdash v : R$} defines an affine function $\R^r \to \R$. In order to reduce $(v \eq w, \psi)$, we consider an independent random variable $X \sim \psi$ and define the auxiliary real random variable $Z = v(X)-w(X)$. If $0$ lies in the support of $Z$, we denote by $\psi|_{v=w}$ the outcome of conditioning $X$ on $Z=0$, and reduce
	\[ (v \eq w, \psi) \red ((), \psi|_{v=w}) \]
	Otherwise $(v \eq w, \psi) \red \bot$, indicating that the inference problem has no solution. To be completely precise, since $v$ and $w$ are affine, the function $(v-w)$ is affine too, so we can find $u \in \R^{1 \times r}$ and $b \in \R$ such that $(v(x)-w(x)) = ux + b$ and then we condition on $uX=-b$ using formula \eqref{eq:conjugacy_simple}.
	\item Let bindings are standard
	\begin{align*} 
	(\letin x v e, \psi) &\red (e[v/x], \psi) \\
	(\letin {(x,y)} {(v,w)} e, \psi) &\red (e[v/x,w/y], \psi)
	\end{align*}
	
	\item Lastly, under reduction contexts, if $(\rho,\psi) \red (e,\psi')$ we define $(C[\rho],\psi) \red (C[e], \psi')$. If $(\rho,\psi) \red \bot$ then $(C[\rho], \psi) \red \bot$.
\end{enumerate}

\begin{proposition}\label{prop:evaluation}
  For every closed typed program $\vdash e : \tau$ either there is a unique value configuration $(v,\psi)$ such that
  $(e,())\red^* (v,\psi)$ with $v$ a value, or $(e,())\red^*\bot$. (Here $()$ is the unique $0$-dimensional Gaussian distribution, and $\red^*$ is the reflexive transitive closure of $\red$.)
\end{proposition}
\begin{proof}[Proof notes]
  First, the $\red$ relation is deterministic, and satisfies progress and type preservation lemmas.
  These are all shown by induction on typing derivations.
  Next, all reduction sequences terminate, because the number of steps is bounded by the number of symbols from $\{\normal,  \eq, \letin{}{}\}$ in an expression.
  \end{proof}
We consider the observable result of this execution either failure, or the pushforward distribution $v_*\psi$ on $\R^n$, as this distribution could be sampled from empirically.

\begin{example}
	The program
	\[ \letin {(x,y)} {(\normal(),\normal())} x \eq y; x+y \]
	reduces to $(z_1 + z_2, \psi)$ where 
	\[ \psi = \mathcal N\left(\begin{pmatrix} 0 \\ 0 \end{pmatrix}, \begin{pmatrix} 0.5 & 0.5 \\ 0.5 & 0.5 \end{pmatrix} \right) \]
	The observable outcome of the run is the pushforward distribution $(1\, 1)_*\psi = \mathcal N(0,2)$ on $\R$.
\end{example}

One goal of this paper is to study properties of this language compositionally, and abstractly, without relying on any specific properties of Gaussians. From the operational semantics, we can define an extensional and compositional contextual equivalence.

\begin{definition}\label{def:ctxequiv}
	We say $\Gamma \vdash e_1, e_2 : \tau$ are \emph{contextually equivalent}, written $e_1 \approx e_2$, if for all closed contexts $K[-]$ and $i,j \in \{1,2\}$
	\begin{enumerate}
		\item when $(K[e_i], !) \red^* (v_i, \psi_i)$ then $(K[e_j], !) \red^* (v_j, \psi_j)$ and $(v_i)_*\psi_i = (v_j)_*\psi_j$
		\item when $(K[e_i], !) \red^* \bot$ then $(K[e_j],!) \red^* \bot$
	\end{enumerate} 
	Here $(v_i)_*\psi_i$ denotes the pushforward distribution as defined in \eqref{def:pushforward}.
\end{definition}

We later study contextual equivalence by developing a denotational semantics for the Gaussian language (\Cref{sec:denotational}), and proving it fully abstract (\Cref{prop:full_abstraction}). We also note nothing conceptually limits the language in this section to only Gaussians. We are running with this example for concreteness, but any family of distributions which can be sampled and conditioned can be used. So we will take care to establish properties of the semantics in a general setting.

\section{Categorical Semantics for Probabilistic Programming}\label{sec:synth_conditioning}

This section introduces denotational semantics for the probabilistic language in \Cref{sec:gaussian_language}, at first without conditioning. We will construct this semantics using a general and reusable strategy, namely
\begin{enumerate}
\item understanding the structure of mathematical models of probability theory, and
\item introducing a metalanguage (the CD calculus) which acts as the internal language of these categorical models
\end{enumerate} 
The Gaussian language will be a particular instance of the metalanguage, and its semantics will take place in the Markov category $\gauss$ (defined in \Cref{def:gauss}). The language for finite probability will have semantics in the Markov category $\finstoch$ (\Cref{def:finstoch}).

This is a general section on the relationship between probabilistic languages and categorical models. We will recall symmetric monoidal categories and string diagrams, and define CD and Markov categories, which are the relevant categorical models of probability theory. We then introduce the CD calculus as a metalanguage for first-order probabilistic programs, and prove the desired correspondence
\begin{enumerate}
\item every program term can be interpreted as a morphisms in a CD category
\item every string diagram can be encoded as a program term
\item every valid manipulation of string diagrams translates to a provable equality of programs, and vice versa
\end{enumerate}
This is formally stated and proven in terms of internal language and syntactic category (\Cref{prop:cd_soundness,def:syncat}). This unlocks the equivalent formalisms discussed in the introduction (\Cref{fig:composition}). These formalisms form the basis of our study of conditioning in the later \cref{sec:abstract_cond,sec:cond}.

\paragraph{Monoidal Categories and String Diagrams}
Recall that a category comprises objects and morphisms between the objects. In this context, the objects are to be thought of as generalized spaces, and the morphisms as stochastic functions. That is, a morphism $X\to Y$ is thought of roughly as something that takes an argument from~$X$ and makes some random choices before returning an element of~$Y$. (The reader familiar with probability theory can regard them as probability kernels, or parameterized measures). In particular, our categories will be monoidal, which means we have the following constructions (see e.g.~\cite{maclane} for full definitions):
\begin{itemize}
\item Monoidal structure: There is a distinguished object $I$ (thought of as the one-point space),
  and for any objects $X$ and $Y$ there is an object $X\otimes Y$ (thought of as the product space). The morphisms $I\to X$ are thought of as probability distributions on $X$, and so the morphisms $Y\to X$ can be thought of as distributions on $X$ with parameters from $Y$. 
\item Categorical composition: for any morphisms $f\colon X\to Y$ and $g\colon Y\to Z$, there is a composite morphism $g \circ f:X\to Z$. This represents running stochastic computation \emph{in sequence}. Mathematically, in several of the examples (Section~\ref{sec:examples-cd}), composition is calculated by a form of integration or summation (integrating over $Y$), and so we can regard this composition as an abstract account of integration. We will often abbreviate the composition $g \circ f$ as $gf$. In the category of sets and functions, morphisms $x : 1 \to X$ can be identified with elements $x \in X$. If $f : X \to Y$ is a function, the composite $f \circ x = fx$ agrees with function application $f(x)$. This notation will be convenient in \Cref{sec:abstract_cond} when applied to deterministic states. 
\item Monoidal composition: for any morphisms $f\colon A\to B$ and $g\colon X\to Y$, there is a composite morphism $(f\otimes g): A\otimes X\to B\otimes Y$. This can be understood as running stochastic computations \emph{in parallel}: informally, given a pair $(a,x)$ we randomly produce $b$ from $a$ and $y$ from $x$, returning the pair $(b,y)$. Mathematically, in the examples, (Section~\ref{sec:examples-cd}), this monoidal composition amounts to product probabilities. In the measure theoretic example (Def.~\ref{def:borelstoch}), the interchange law of the tensor $\otimes$ encodes Fubini's theorem (see \Cref{app:meas}).
\end{itemize}

The notation for composing morphisms (with $\circ$ and $\otimes$) can quickly become cryptic, and it is important to find good notations. \emph{String diagrams} are one such widely used and intuitive notation (e.g.~\cite{joyal_street_1,selinger:graphical}). In the string diagram, the objects of the category become wires and morphisms are boxes; sequential composition ($\circ$) is simply joining wires together, and monoidal composition ($\otimes$) is juxtaposition. States ($I \to X$) and effects ($X \to I$) have special notations, because the unit object $I$ need not be drawn in string diagrams (see \Cref{fig:stringdiagrams} for an overview).

\begin{figure}[h]
\[ \input{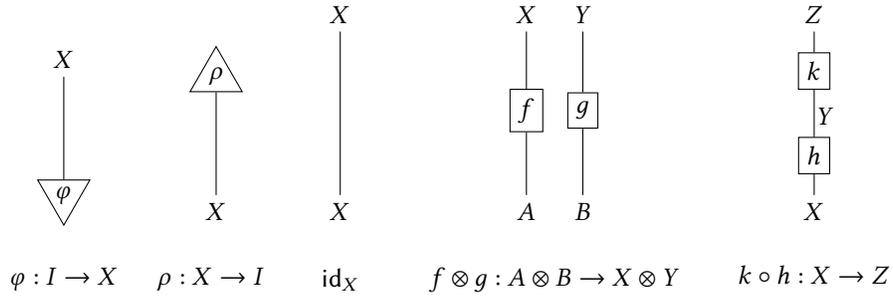} \]
\caption{Overview of string diagram notation}
\label{fig:stringdiagrams}
\end{figure}

We read such diagrams from bottom to top. Even without categorical machinery, string diagrams carry an intuitive meaning as dataflow diagrams. We can manipulate string diagrams using intuitive rules (sliding around of boxes, formally: ``planar isotopy''), and the axioms of monoidal categories ensure that all such manipulations result in the same overall composite. Even more, if there are different ways of parsing a string diagram into a sequence of composites, then all these ways have provably equal meanings. This result is known as \emph{coherence} for monoidal categories. We demonstrate this for the so-called \emph{interchange law} $(g \circ f) \otimes (g' \circ f') = (g \otimes g') \circ (f \otimes f')$ which is derivable from the axioms of monoidal categories. It corresponds to the unambiguous reading of the string diagram

\[ \input{categorical_probability/interchange.tikz}  \]

\paragraph{Symmetry, copying and discarding}

Monoidal categories are an important general concept, but to discuss probabilities in this axiomatic way, it is appropriate to require further structure, resulting in copy-delete categories and Markov categories as we discuss in Section~\ref{sec:cd_cats}. The crucial operations are swapping, copying and discarding of data, which we depict as follows 

\[ \input{categorical_probability/comon_structures_simple.tikz} \]

The same building blocks are necessary to interpret terms of a programming language in which variables can be used without linearity restriction. For example, the term-in-context $x,y,z \vdash (y,x,x)$ requires us to swap and copy $x$, as well as discard $z$. The same could be achieved by first copying $x$ and then swapping both copies with $y$. These two diagrams are provably equal from the axioms of CD categories,

\[ \input{categorical_probability/copyswap.tikz} \]

In defining the CD-calculus (\Cref{sec:cd-calc}), we extend the syntax with let-binding, tuples, projections and function calls. The calculus has equivalent expressive power to string diagrams, as showcased in the introduction (\Cref{fig:composition}). In \Cref{sec:cd-semantics}, we give the systematic method to associate to every term $t$ a string diagam $\sem{t}$. In particular, manipulations of the string diagrams correspond to valid program transformations. For example, using the definition of the semantics of let-bindings in \Cref{fig:cd_semantics}, the validity of the commutativity equation \eqref{eqn:initcommutativity} corresponds to the following string diagram manipulation (where $x \notin \fv(t), y \notin \fv(s)$)

\[ \input{categorical_probability/letexample.tikz} \]

In these opening remarks so far, we have started from categorical notions and moved to notations. It is also helpful to follow the opposite route: we can regard notation as primal -- be it string diagrams, graphical models, or probabilistic programs. Now to decide whether two composites are equal (two diagrams, two programs) we regard them as morphisms in a category (called the \emph{syntactic category}), and ask whether they are equal there (\Cref{sec:cd-semantics}). In this way, category theory is merely a formalism for compositional theories of equality, which are useful from a foundational perspective as well as for understanding valid program manipulations. We axiomatize this equality from the programming perspective in Section~\ref{cd:equations}. \\

We will now formally introduce CD categories and Markov categories and define the relevant mathematical models for finite and Gaussian probability, before returning to the CD calculus in \Cref{sec:cd-calc}.

\subsection{Copy-Delete Categories and Markov Categories}\label{sec:cd_cats}

We will recall two closely related notions, namely:

\begin{description}
\item[Markov categories] to model purely stochastic computation \citep{fritz} and
\item[CD categories] which model potentially \emph{unnormalized} stochastic computation \citep{cho_jacobs}.
\end{description}
Markov categories have been used to formalize various theorems of probability and statistics, such as sufficient statistics (Fisher-Neyman, Basu, Bahadur) \citep{fritz}, stochastic dominance (Blackwell-Sherman-Stein) \citep{representable_markov} and zero-one laws \citep{fritz:zero_one}. Both types of category admit a convenient graphical language in terms of string diagrams. We will also define the CD-calculus, which is the internal language of CD categories and reminiscent of first-order OCaml  (\Cref{sec:cd-calc}). This makes CD categories a natural foundation for probabilistic programming. Denotational semantics will be given to our Gaussian language by recognizing it as the internal language of an appropriate CD category. 

\begin{definition}[CD category, {\cite{cho_jacobs}}]
  \label{def:cd}
	A \emph{copy-delete category} (CD category) is a symmetric monoidal category $(\C, \otimes, I)$ where every object $X$ is equipped with the structure of a commutative comonoid
	\[ \cpy_X : X \to X \otimes X \quad \del_X : X \to I \]
	graphically depicted as
	\[ \input{categorical_probability/comon_structures.tikz} \]
	satisfying the axioms
	\[ \input{categorical_probability/comonoid1.tikz} 
	\qquad\quad\input{categorical_probability/comonoid2.tikz} \]
	We require that the comonoid structure be compatible with the monoidal structure as follows
	\[ \input{categorical_probability/multiplicativity.tikz} \]
\end{definition}

It is important that $\cpy$ and $\del$ are \emph{not} assumed to be natural; explicitly the equations 
\begin{equation} \input{categorical_probability/non_copyable.tikz} \label{cd:def_copy_disc} \end{equation}
need \emph{not} hold in general. We give special names to situations where they do hold.

\begin{definition}[Copyable, discardable, deterministic]\label{def:copydisc}
	A morphism $f : X \to Y$ is called \emph{copyable} if the first equation of \eqref{cd:def_copy_disc} holds:
	$ \cpy_Y \circ f = (f \otimes f) \circ \cpy_X$.
	A morphism $f : X \to Y$ is called \emph{discardable} if the second equation of \eqref{cd:def_copy_disc} holds: $\del_Y \circ f = \del_X$. A morphism is called \emph{deterministic} if it is copyable and discardable.
\end{definition}

\begin{definition}[Markov category, {\cite{fritz}}]\label{def:markovcat}
	A Markov category is a CD category $\C$ in which the following equivalent properties hold
	\begin{enumerate}
		\item $\C$ is semicartesian, i.e. the unit $I$ is terminal
		\item every morphism is discardable
		\item $\del$ is natural
	\end{enumerate}
\end{definition}

The definitions of CD- and Markov categories encode a significant amount of properties of stochastic computation, as discussed informally at the beginning of Section~\ref{sec:synth_conditioning}. The discardability condition in Markov categories informally means that probabilities sum to $1$, as is the case in the examples (Section~\ref{sec:examples-cd}). 

\emph{Notation:} The presence of explicit copying and discarding maps lets us apply a product-like syntax for CD categories: If $f : A \to X$, $g : A \to Y$, we write a tupling
\[ \langle f,g \rangle \defeq (f \otimes g) \circ \cpy_A \]
and define projection maps
\[ \pi_X : X \otimes Y \to X \quad \pi_Y : X \otimes Y \to Y \]
via discarding. Recall that a \emph{state} in a symmetric monoidal category is a morphism $\psi : I \to X$. An \emph{effect} is a morphism $\rho : X \to I$. Note that by terminality of the unit $I$ in a Markov category, all effects $X \to I$ must be trivial; in CD categories, effects will be of interest. 

In Markov categories, we will furthermore employ the following probabilistic terminology: We call states $\psi : I \to X$ \emph{distributions}, and if $f : A \to X \otimes Y$, we define its \emph{marginal} (of $X$) $f_X : A \to X$ to be $\pi_Xf$. Of course, we generally have $f \neq \langle f_X, f_Y \rangle$ unless $\C$ is cartesian. For every CD category $\C$, the wide subcategory $\C_\det$, which consists of only the deterministic morphisms, is cartesian \cite[Remark 10.13]{fritz}.

\subsection{Examples of CD categories}\label{sec:examples-cd}
In this subsection, we will briefly introduce the relevant examples of Markov and CD categories we will be working with, in particular finite and Gaussian probability. All examples in this subsection are standard material, and for example covered in \cite{fritz}. More mathematical detail and the theory of conditioning is given in \Cref{sec:abstract_cond}. Readers primarily interested in syntax can proceed with \Cref{sec:cd-calc}.

\begin{definition}[{\cite[2.5]{fritz}}]\label{def:finstoch}
The Markov category $\finstoch$ has as objects finite sets $X$, and morphisms $X \to Y$ are probability channels, that is stochastic matrices $p \in [0,1]^{Y \times X}$, which are sometimes written in the notation $p(y|x)$. Composition in $\finstoch$ takes the form of the \emph{Kolmogorov-Chapman equation}
\begin{equation} (pq)(z|x) = \sum_y p(z|y)q(y|x) \label{eq:chapman} \end{equation}
\end{definition}
We modify $\finstoch$ as follows to allow for unnormalized (`sub-stochastic') computation:
\begin{definition}\label{def:substoch}
The CD category $\finsubstoch$ has as objects finite sets $X$, and morphisms $X \to Y$ are \emph{subprobability} channels $p(y|x)$, that is $p(y|x) \in [0,1]$ and for all $x \in X$,
\[ \sum_y p(y|x) \leq 1 \]
Composition is again given by \eqref{eq:chapman}.
\end{definition}

A convenient way to understand composition in these categories is using the theory of monads.
\begin{definition}
The distribution monad $D$ and subdistribution monad $D_{\leq 1}$ on the category of sets are defined as the sets of all (sub)probability distributions on a given set $X$;
\begin{align*}
D(X) &= \{ p : X \to [0,1] \text{ finitely supported } : \sum_x p(x) = 1 \} \\
D_{\leq 1}(X) &= \{ p : X \to [0,1] \text{ finitely supported } : \sum_x p(x) \leq 1 \} 
\end{align*}
The unit of both monads is given by taking Dirac distribution $x \mapsto \delta_x$ with
\[ \delta_x(y) = \begin{cases} 1, &\text{if } x = y \\
0 &\text{if } x \neq y
\end{cases} \]
and the monad multiplication $\mu$ is defined as
\[ \mu(\rho)(x) = \sum_{p} \rho(p) \cdot p(x) \]
\end{definition}
Kleisli composition for these monads recovers the Kolmogorov-Chapman equation \eqref{eq:chapman}. That is, morphisms in $\finstoch(X,Y)$ are Kleisli arrow $X \to D(Y)$ for $D$, and morphisms in $\finsubstoch(X,Y)$ are Kleisli arrows for $X \to D_{\leq 1}(Y)$. The following result shows that Kleisli categories are a general source of CD categories.

\begin{proposition}[{\cite[3.2]{fritz}}]
Let $\C$ be a category with finite products and $T : \C \to \C$ be a strong, commutative monad. Then the Kleisli category $\kl(T)$ is a CD category, which is furthermore Markov if and only if $T$ is \emph{affine}, i.e. $T1 \cong 1$.
\end{proposition}

Again, a morphism in $\finsubstoch(X,Y)$ is a Kleisli arrow $X \to D_{\leq 1}(Y)$ for the subprobability monad on $\set$. \\

We now define the Markov category which captures the Gaussian probability of Section~\ref{sec:recap_gauss}.

\begin{definition}[{{\cite[\S 6]{fritz}}}]\label{def:gauss}
The Markov category $\gauss$ has objects $n \in \mathbb N$, which represent the affine space $\R^n$, and $m \otimes n = m+n$. Morphisms $m \to n$ are tuples $(A,b,\Sigma)$ where $A \in \R^{n \times m}, b \in \R^n$ and $\Sigma \in \R^{n \times n}$ is a positive semidefinite matrix. The tuple represents a stochastic map $f : \R^m \to \R^n$ that is affine-linear, perturbed with multivariate Gaussian noise of covariance $\Sigma$, informally written
	\[ f(x) = Ax + b + \mathcal N(\Sigma) \text{ or } Ax + \mathcal N(b,\Sigma) \]
	Such morphisms compose sequentially and in parallel in the expected way, with noise accumulating independently
	\begin{align*} (A,b,\Sigma) \circ (C,d,\Xi) &= (AC, Ad + b, A\Xi A^T + \Sigma) \\
	(A,b,\Sigma) \otimes (C,d,\Xi) &= \left( 
	\begin{pmatrix} A & 0 \\ 0 & C \end{pmatrix},
	\begin{pmatrix} b \\ d \end{pmatrix},
	\begin{pmatrix} \Sigma & 0 \\ 0 & \Xi \end{pmatrix} \right)
	\end{align*}
Copy- and discard structure are given using the affine maps
	\[ \cpy_n : \R^n \to \R^{n+n}, x \mapsto (x,x) \quad \del_n : \R^n \to \R^0, x \mapsto () \]
\end{definition}

Note that we explicitly allow zero covariance in the definition of $\gauss$. This way, the category is able to encode deterministic computation as a special case. 

\begin{proposition}
	A morphism $(A,b,\Sigma)$ in $\gauss$ is deterministic (\Cref{def:copydisc}) iff $\Sigma = 0$, i.e. there is no randomness involved.
\end{proposition}
\begin{proof}
	Write $f = (A,b,\Sigma)$, then the covariance matrices of $f \circ \cpy$ and $\cpy \circ f$ are
	\[ \begin{pmatrix} \Sigma & 0 \\ 0 & \Sigma \end{pmatrix} \text{ and } \begin{pmatrix} \Sigma & \Sigma \\ \Sigma & \Sigma \end{pmatrix} \]
	respectively. Thus $f$ is copyable iff $\Sigma = 0$. 
\end{proof}

It follows that the deterministic subcategory $\gauss_\det$ is the category $\catname{Aff}$ consisting of the spaces $\R^n$ and affine maps between them. \\

Note that this definition of $\gauss$ involves no measure theory at all; a Gaussian is fully described by its mean and covariance matrix. Measure theory can however be used to build a Markov category that is rather comprehensive in that it includes the previous two examples. We briefly state the definition here and refer to the appendix (\Cref{app:meas}) for details. (In this paper, we will not use this measure-theoretic Markov category in a crucial way, we will only use it in illustrative examples.)

\begin{definition}[{{\cite[\S4]{fritz}}}]\label{def:borelstoch}
The Markov category $\borelstoch$ has as objects standard Borel spaces $X$, and morphisms $X \to Y$ are probability kernels $\Sigma_X \times Y \to [0,1]$. 
\end{definition}	

$\borelstoch$ arises as the Kleisli category of the Giry monad $\G$ on standard Borel spaces \cite{giry}. Both $\finstoch$ and $\gauss$ are subcategories of $\borelstoch$, i.e. there are faithful inclusion functors which preserve all CD structure. 

Lastly, we give an example to show that the formalism of CD categories can not only encompass probabilistic situations but also nondeterminism.
\begin{definition}[{{e.g.~\cite[Ex.~2.6 \& \S8.1]{fritz}}}]
We denote by $\rel$ the CD category of sets $X,Y$ and relations $R \subseteq X \times Y$ between them. We denote by $\rel^+$ the Markov subcategory of sets and left-total relations between them, i.e. $\forall x \in X \exists y \in Y, (x,y) \in R$.
\end{definition}
The two categories are obtained as the Kleisli categories of the powerset monad $\mathcal P : \set \to \set$ and the nonempty powerset monad $\mathcal P^+: \set \to \set$ respectively. The category $\rel^+$ is referred to as $\catname{SetMulti}$ in \cite{fritz}.

\subsection{Internal Languages and Denotational Semantics}
\label{sec:cd-calc}
We present the CD calculus, which is the internal language CD categories. It is reminiscent of the first-order fragment of fine-grained call-by-value or the computational $\lambda$-calculus (e.g.~\cite{moggi:computational_lambdacalculus,moggi:computation_and_monads,levy:fcbv}), but the commutativity of the tensor allow for some convenient simplifications and a concise equational presentation. To this extent it is a novel calculus. 

\begin{definition} A CD signature $\mathfrak S = (\tau,\omega)$ consists of sets $\tau$ of base types and function symbols $\omega$. A \emph{type} is recursively defined by closing the base types under tuple formation
\[ A ::= \tau \s \tunit \s A \ast A \]
Each function symbol $f \in \omega$ is equipped with a unary \emph{arity} of types, written $f : A \to B$. The terms of the CD-calculus are given by
\[ t ::= x \s () \s (t,t) \s \pi_i\,t \s f\,t \s \letin x t t \quad\quad (i=1,2) \]
subject to the typing rules $x_1 : A_1, \ldots, x_n : A_n \vdash t : B$ given in \Cref{fig:cd_types}.
\end{definition}
\begin{figure}[h]
	\centering
	\begin{align*}&
	\infer{\Gamma, x : A, \Gamma' \vdash x : A}{} \qquad
	\infer{\Gamma \vdash () : \tunit}{}
	\qquad \infer{\Gamma \vdash (s,t) : A \ast B}{\Gamma \vdash s : A \quad \Gamma \vdash t : B} \qquad
	\infer[(f : A \to B)]{\Gamma \vdash f\,t : B}{\Gamma \vdash t : A}
        \\[6pt]
	&\infer{\Gamma \vdash \pi_1\,t : A_1}{\Gamma \vdash t : A_1 \ast A_2} \qquad \infer{\Gamma \vdash \pi_2\,t : A_2}{\Gamma \vdash t : A_1 \ast A_2} \qquad
	\infer{\Gamma \vdash \letin x e t : B}{\Gamma, x : A \vdash t : B \quad \Gamma \vdash e : A}
	\end{align*}
	\caption{Typing rules for the CD calculus}
	\label{fig:cd_types}
\end{figure}

We employ some standard syntactic sugar, for example sequencing 
\[ s;t \defeq \letin {x} s t \quad\quad (x \notin \fv(t)) \]
We also define a pattern-matching let as syntactic sugar
\[ (\letin {(x,y)} s t) \defeq (\letin p s \letin x {\pi_1\,p} \letin y {\pi_2\,p} t) \text.\]
Conversely, we can provably recover the projection constructs from this sugar (in a sense made precise by the equational theory in \Cref{cd:equations}):
\[
(\pi_1\,s) = (\letin {(x,y)} s x)\qquad 
(\pi_2\,s) = (\letin {(x,y)} s y) 
\]
We prefer the projections over pattern-matching when presenting the equational theory, because this means one less binding construct.

\subsection{Semantics}
\label{sec:cd-semantics}
We now explain how the CD calculus can be interpreted in CD categories. Types will be interpreted as objects, and terms interpreted as morphisms. Formally, a \emph{model} of signature $(\tau,\omega)$ is a CD category $\C$ together with an assignment of objects $\sem{A} \in \C$ for each basic type and morphisms $\sem{f} : \sem{A} \to \sem{B}$ for each function symbol $f : A \to B$. Here we extend $\sem{-}$ to arbitrary types and contexts by
\[ \sem{\tunit} = I \quad \sem{A_1 \ast A_2} = \sem{A_1} \otimes \sem{A_2} \quad \sem{A_1, \ldots, A_n} = \sem{A_1} \otimes (\cdots \otimes \sem{A_n}) \]

For any model, the interpretation of a term $\Gamma \vdash t : A$ is defined recursively as 
\begin{itemize}
	\item $\sem{x}$ is the discarding map $\sem{\Gamma,A,\Gamma'} \cong \sem{\Gamma} \otimes \sem{A} \otimes \sem{\Gamma'} \to I \otimes \sem{A} \otimes I \cong \sem{A}$
	\item $\sem{()}$ is the discarding map $\del_{\sem{\Gamma}} : \sem{\Gamma} \to I$
	\item $\sem{(s,t)}$ is the map $\sem{\Gamma} \xrightarrow{\cpy_{\sem{\Gamma}}} \sem{\Gamma} \otimes \sem{\Gamma} \xrightarrow{\sem{s} \otimes \sem{t}} \sem{A} \otimes \sem{B} = \sem{A \ast B}$
	\item $\sem{\pi_i t}$ is marginalization $\sem{\Gamma} \xrightarrow{\sem{t}} \sem{A_1} \otimes \sem{A_2} \to \sem{A_i}$
	\item $\sem{f\,t}$ is the composite $\sem{\Gamma} \xrightarrow{\sem{t}} \sem{A} \xrightarrow{\sem{f}} \sem{B}$
	\item $\sem{\letin x e t}$ is given by $\sem{\Gamma} \xrightarrow{\cpy_{\sem{\Gamma}}} \sem{\Gamma} \otimes \sem{\Gamma} \xrightarrow{\id_{\sem{\Gamma}} \otimes \sem{e}} \sem{\Gamma} \otimes \sem{A} \xrightarrow{\sem{t}} \sem{B}$
\end{itemize}

The semantics can be seen as a procedure for translating every term of the CD calculus into a string diagram, as shown in \Cref{fig:cd_semantics}.

\begin{figure}[h]
\[ \input{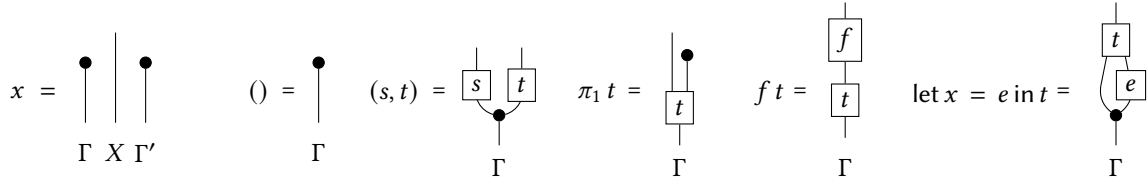} \]
\caption{Translating terms into string diagrams (brackets $\sem{-}$ omitted for readability)} \label{fig:cd_semantics}
\end{figure}

\begin{proposition}[Structural rules for contexts]
	\label{cd:structural}
	The weakening and exchange rules 
	\[ \infer{\Gamma, x : A\vdash t : B}{\Gamma \vdash t : B} \qquad \infer{\Delta,\Gamma \vdash t : B}{\Gamma, \Delta \vdash t : B} \]
	are derivable, and their semantics corresponds to precomposition with the discard and swap morphisms.
\end{proposition}
\begin{proof}
	Straightforward induction using the comonoid axioms.
\end{proof}

As an example, we use the CD calculus to give straightforward denotational semantics to the conditioning-free fragment of the Gaussian language in $\gauss$. We notice that this fragment is precisely the CD calculus for the signature $\mathfrak S$ with base type $\rv$ and function symbols
\[(+) : \rv \ast \rv \to \rv \qquad
\beta \cdot (-) : \rv \to \rv \qquad
\underline \beta : \tunit \to \rv \qquad
\normal : \tunit \to \rv \]

The Markov category $\gauss$ models this signature using $\sem{\rv} = 1$ and the obvious interpretations of the function symbols. \\

The goal of \Cref{sec:cond} will be to interpret the full Gaussian language in a CD category $\cond(\gauss)$. That category will need to interpret the additional function symbol $(\eq) : \rv \ast \rv \to \rv$.

\subsection{Equational Theory}\label{cd:equations}
We now give a sound and complete equational theory with respect to CD models. \\

In call-by-value languages, the substitution 
\[ (\letin x e u) \equiv u[e/x] \]
is generally only admissible if $e$ is a \emph{value expression}, that is it does not produce effects. In the CD calculus, another powerful substitution scheme is valid: We can replace $(\letin x e u) \equiv u[e/x]$ whenever $u$ uses $x$ \emph{linearly}, i.e. exactly once, \emph{even if} $e$ is an effectful computation. Using the linear- and value substitution schemes, the theory of the CD calculus can be presented concisely as in \Cref{fig:cd_axioms}. Note that we omit the context of equations when unambiguous and identify bound variables up to $\alpha$-equivalence. Whenever we say ``use'' or ``occurrence'', we mean free use and occurrence, and substitution is always capture-avoiding.  
\begin{figure}[h]\framebox{\begin{minipage}{0.98\linewidth}
  \raggedright
  Congruence laws:
	\begin{gather}
	\text{$\equiv$ is reflexive, symmetric and transitive} \label{cd:equiv} \tag{equiv} \\
	\infer{(\letin x {e_1} e_2) \equiv (\letin x {e_1'} e_2')}{e_1 \equiv e_1' \quad e_2 \equiv e_2'} \label{cd:let_cong} \tag{let.$\xi$}
        \end{gather}
A \emph{value expression} is a term of the form 
	\[ V ::= x \s () \s (V,V) \s \pi_i V \s \letin x V V \]
	The axioms of the CD calculus are:
        \begin{align}
      	(\letin x e t) &\equiv t[e!x] \label{cd:let_lin} \tag{let.lin} \\
	(\letin x V t) &\equiv t[V/x] \label{cd:let_val} \tag{let.val} \\
	\pi_i\,(x_1,x_2) &\equiv x_i \label{cd:pair_beta} \tag{$\ast$.$\beta$} \\
	(\pi_1\,x, \pi_2\,x) &\equiv x \label{cd:pair_eta} \tag{$\ast$.$\eta$} \\
	x &\equiv () \label{cd:unit_eta} \tag{$\tunit$.$\eta$} 
	\end{align}
	where we write $t[e ! x]$ for substituting a unique free occurrence of $x$. For the internal language of \emph{Markov categories}, extend \eqref{cd:let_lin} to all substitutions targeting \emph{at most one} free occurrence of $x$.\end{minipage}}
	\caption{Axioms of the CD-calculus}
	\label{fig:cd_axioms}
\end{figure}

\begin{proposition}[Soundness]\label{prop:cd_soundness}
	Every CD model validates the axioms of the CD calculus. That is if $\Gamma \vdash e_1 \equiv e_2 : A$ then $\sem{e_1} = \sem{e_2} : \sem{\Gamma} \to \sem{A}$.
\end{proposition}
\begin{proof}
	The proofs are straightforward if tedious string diagram manipulations. We showcase the validation of one interesting equation, \eqref{cd:assoc}, here and move the remaining derivations to the appendix (\Cref{app:cdcalc}). Let $\Gamma \vdash e_1 : X_1$, $\Gamma, x_1 : X_1 \vdash e_2 : X_2$ and $\Gamma, x_2 : X_2 \vdash e : Y$. Then showing
	\[ \sem{(\letin {x_1} {e_1} \letin {x_2} {e_2} e^{\mathsf w})} \equiv \sem{(\letin {x_2} {(\letin {x_1} {e_1} e_2)} e)} \]
	translates to the following manipulation of string diagrams
	\begin{equation}
	\input{categorical_probability/eq_assoc.tikz} \label{cd:eq_assoc}
	\end{equation}
	Note that we formally write $e^{\mathsf w}$ to be fully explicit about weakening $e$; its denotation discards the unused $X_1$-wire as per \Cref{cd:structural}.
\end{proof}

The equational theory lets us derive many useful program equations, including commutativity. 

\begin{proposition}\label{prop:all_lambdac_axioms}
	All axioms of the ground $\lambda_c$-calculus \citep[Tables~6,7]{moggi:computation_and_monads} and commutativity are derivable. 
	\begin{align}
	(\letin {x_2} {(\letin {x_1} {e_1} e_2)} e) &\equiv (\letin {x_1} {e_1} \letin {x_2} {e_2} e) \quad x_1 \notin \fv(e) \label{cd:assoc} \tag{assoc} \\
	(\letin {x_1} {e_1} \letin {x_2} {e_2} e) &\equiv (\letin {x_2} {e_2} \letin {x_1} {e_1} e) \quad x_1 \notin \fv(e_2), x_2 \notin \fv(e_1) \label{cd:comm}~\tag{comm} \\
	(\letin x e x) &\equiv e \label{cd:id} \tag{id} \\
	(\letin {x_1} {x_2} e) &\equiv e[x_2/x_1] \label{cd:let_beta} \tag{let.$\beta$} \\
	f e &\equiv (\letin x e f x) \label{cd:let_f} \tag{let.f}\\
	(s,t) &\equiv (\letin x s \letin y t (x,y)) \label{cd:let_ast} \tag{let.$\ast$} 
	\end{align}
	Note that by commutativity \eqref{cd:comm}, the order of evaluation in \eqref{cd:let_ast} does not matter. 
\end{proposition}
\begin{proof}In the appendix (\Cref{app:cdcalc}). \end{proof}

We proceed with some syntactic remarks about the CD calculus on the relationship between linear substitution to general nonlinear substitutions: If $t$ is a term with $n$ free occurrences of the variable $x$, let $\hat t$ denote the term $t$ with those occurrences replaced with distinct fresh variables $x_1, \ldots, x_n$ (the order does not matter). By repeated application of \eqref{cd:let_val}, we can derive
\begin{equation}
t \equiv  \letin {x_1} x \cdots \letin {x_n} x \hat t 
\end{equation}
We can now substitute some or all occurrences of $x$ using \eqref{cd:let_lin} as follows
\begin{equation}
t[e/x] \equiv \hat t[e!x_1] \cdots [e!x_n] \equiv \letin {x_1} e \cdots \letin {x_n} e \hat t \label{cd:general_subs}
\end{equation}
This means we can reduce questions about substitution to the copying behavior of the term $e$. We adapt the definitions from \citep{fuhrmann,kammarplotkin}.
\begin{definition}
	A term $e$ is called \emph{copyable} if
	\begin{equation} (\letin x e (x,x)) \equiv (e,e) \label{cd:copyable} \end{equation}
	is derivable. A term $e$ is called \emph{discardable} if
	\begin{equation} (\letin x e ()) \equiv () \label{cd:discardable} \end{equation}
	is derivable. We call $e$ \emph{deterministic} if it is both copyable and discardable.
\end{definition}

\begin{proposition} \label{prop:cd_subst}
	The substitution equation
	\[ (\letin x e t) \equiv t[e/x] \]
	is derivable in any of the following circumstances:
	\begin{enumerate}
		\item $t$ uses $x$ exactly once
		\item $t$ uses $x$ at least once, and $e$ is copyable
		\item $t$ uses $x$ at most once, and $e$ is discardable
		\item $e$ is deterministic (combining the previous two points)
	\end{enumerate}
\end{proposition}

Finally, we remark that the CD calculus is complete with respect to CD models. We employ the usual construction of a \emph{syntactic category} or free CD category over a given CD signature. Not only can every term be translated into a string diagram, also every string diagram can be parsed into a term, and the theory of $\equiv$ proves all ways of reading a diagram equivalent. 

\begin{definition}\label{def:syncat}
Fix a CD signature $\mathfrak S$. The \emph{syntactic category} $\catname{Syn}$ has
	\begin{enumerate}
		\item objects are types $A$
		\item morphisms are equivalence classes of terms $x : A \vdash t : B$ modulo $\equiv$
		\item identities are variables $x : A \vdash x : A$
		\item composition is let binding; if $x : A \vdash s : B$ and $x : B \vdash t : C$, their composite is
		\[ x : A \vdash \letin x e t \]
		\item Tensor on objects is defined as $A \otimes B = A \ast B$ with unit type $\tunit$. The tensor on morphisms of $x_1 : A_1 \vdash s_1 : B_1$, $x_2 : A_2 \vdash s_1 : B_2$ is
		\[ x : A_1 \ast A_2 \vdash \letin {x_1} {\pi_1\,x} \letin {x_2} {\pi_2\,x} (s_1,s_2) \]
		\item CD structure is given by nonlinear use of variables, that is
		\begin{align*}
		\cpy_A &= x : A \vdash (x,x) : A \ast A \\
		\del_A &= x : A \vdash () : \tunit
		\end{align*}
	\end{enumerate}
\end{definition}
The verification of the CD category axioms is tedious but standard. Note that we can build on existing work \cite{levy:fcbv} because our axioms prove all equations of the ground fragment of $\lambda_c$ (\Cref{prop:all_lambdac_axioms}).
We expect that the syntactic category is an initial model over a given signature and the definition of the semantics $\sem{-}$ is forced by preserving CD structure, but we won't formalize this here.

\section{An Abstract Account of Inference}\label{sec:abstract_cond}

In \Cref{sec:synth_conditioning}, we recalled Markov categories (\Cref{def:markovcat}) as abstract formulations of probability theory, equipped with multiple notational formalisms: string diagram as well as programming notations. We now present an abstract theory of \emph{inference problems} in Markov categories with sufficient structure. We approach the topic from the programming languages side. An informal outline is as follows: Roughly, an \emph{inference problem} is a closed program of the form 

\begin{equation} \letin {(x,k)} \psi (k \eqo o); \return x
\qquad \qquad\text{for some $\psi$ and $o$}
  \label{eq:closed_inf} \end{equation}
where $k \eqo o$ is, at this point, simply a notation for recording an exact condition we wish to make --- that the second marginal $k$ of $\psi$ is equal to the observation~$o$ (see Sec.~\ref{sec:inferenceproblems}). 
A \emph{solution} to this problem is a contextually equivalent program which no longer mentions $k\eqo o$. 
Our treatment is now guided by program equations like \eqref{eqn:substlong} and \eqref{eqn:initlong}, and also the symbolic approach of \cite{shan_ramsey}: Finding a conditional for~$\psi$ amounts to restructuring the dataflow of \eqref{eq:closed_inf} in a way that the return value $x$ is expressed in terms of the observation $k$ (where $\psi_K = \pi_2\,\psi$):
\[ \letin k {\psi_K} \letin x {\psi|_K(k)} (k \eqo o); \return x
\qquad
\text{for some $\psi|_K$}\text.\]
By commutativity~\eqref{eqn:initcommutativity} and \eqref{cd:comm}, this is the same as
\[ \letin k {\psi_K} (k \eqo o); \return \psi|_K(k) \]
We can regard the initialization principle~(\ref{eqn:initlong}) as a fundamental property of exact conditions ($k\eqo o$), and then use this. Informally, if it is possible for $k$ to be $o$ ($k \ll \psi_K$, Def.~\ref{def:abscont}), we may simply substitute the observation for the variable $k$,
\[ \letin k o \psi|_K(k)\]
which finally results in the solution $\psi|_K(o)$, which is equivalent to~(\ref{eq:closed_inf}). If on the other hand it is impossible for $k$ to be equal to $o$ ($k \not\ll \psi_K$), then the inference problem has no solution and is infeasible.

For the rest of this chapter, we formally express this idea of conditioning via program transformation in terms of Markov categories. We begin by recalling categorical rephrasings of core notations from measure-theoretic probability: conditional probability (\Cref{sec:conditionals}), and almost-sure equality, absolute continuity and support (\Cref{sec:as_equality}), mostly due to~\cite{cho_jacobs,fritz}.
We then formulate precisely what an inference problem is, and when it succeeds (\Cref{sec:inferenceproblems}). For now, we summarize informally: an inference problem is a pair $(\psi,o)$ of a distribution $\psi$ on a compound space~$X\otimes K$ and a deterministic observation~$o$ about $K$, regarded in the spirit of~(\ref{eq:closed_inf}); it succeeds if we are able to infer a conditional distribution, or posterior, about $X$, and fails otherwise. This failure is not sought in practice, but happens for instance if one attempts to record two different exact observations about the same data point, or in general if the observation~$o$ is outside the support of~$\psi$. 

Looking forward, we will develop a notion of \emph{open inference problem} in \Cref{sec:cond} as part of a compositional framework for collecting conditions, so as to give a compositional semantics to the kind of programming with conditioning demonstrated in \Cref{sec:gaussian_language}.

For the rest of this section, we fix a Markov category $\C$ (Def.~\ref{def:markovcat}).

\subsection{Conditionals}\label{sec:conditionals}

In essence, conditioning is a way of recovering a joint distribution only given access to part of its information. Given a joint distribution over $(X,Y)$, we can always form a generative story where the value of $X$ is sampled first, and then $Y$ is computed depending (or conditional) on $X$. The categorical formulations of conditioning trace back to Golubtsov and Cho-Jacobs \cite{golubtsov2002monoidal,cho_jacobs}.

\begin{definition}[{{\cite[11.1]{fritz}}}]\label{def:conditional} A \emph{conditional distribution} for $\psi : I \to X \otimes Y$ (given $X$) is a morphism $\psi|_X : X \to Y$ such that
	\begin{equation} \input{categorical_probability/cond_dist.tikz} \label{eq:cond_dist} \end{equation}
	More generally, a \emph{(parameterized) conditional} for $f : A \to X \otimes Y$ is a morphism $f|_X : X \otimes A \to Y$ such that
	\begin{equation} \input{categorical_probability/cond_param.tikz} \end{equation}
\end{definition}

It is worth spelling out that \eqref{eq:cond_dist}, expressed in the CD-calculus, corresponds precisely the type of restructuring of dataflow which was discussed in the introduction to this chapter. The equation~\eqref{eq:cond_dist} simply becomes
\[ \psi \quad \equiv \quad (\letin x {\pi_1\,\psi} (x,\psi|_X(x))) \]

Parameterized conditionals can again be specialized to conditional distributions by fixing a parameter
\begin{proposition}
	\label{prop:specialization}
	If $f : A \to X \otimes Y$ has conditional $f|_X : X \otimes A \to Y$ and $a : I \to A$ is a \emph{deterministic} state, then $f|_X(\id_X \otimes a)$ is a conditional distribution for the composite $f \circ a$.
\end{proposition}
\begin{proof}
	Using determinism of $a$, we check that
	\[ \input{categorical_probability/specialization.tikz} \]
\end{proof}

\begin{proposition}\label{prop:ex_conditionals}
	$\finstoch$, $\borelstoch$, $\gauss$ and $\rel^+$ have all conditionals.
\end{proposition}
\begin{proof}
	In $\borelstoch$, the definition of conditionals instantiates to \emph{regular conditional distributions} which are known to exist under the assumptions of the category (that is on \emph{standard Borel spaces}) (see \citep[11.7]{fritz}, \citep[Thm.~3.5]{bogachev}). As a special case, conditionals in $\finstoch$ are given by the traditional conditional distribution \citep[11.2]{fritz}
	\begin{equation} \psi|_X(y|x) = \frac{\psi(x,y)}{\psi_X(x)} \quad \text{ when } \pi_X(x) > 0 \label{eq:finstoch_conditional} \end{equation}
	Conditionals in $\gauss$ exist and can be given using an explicit formula generalizing \eqref{eq:conjugacy_formula} \cite[11.8]{fritz}. The property that conditionals of Gaussians are again Gaussian is sometimes called \emph{self-conjugacy} \citep{jacobs:conjugatepriors}.
	
	In $\rel^+$, the conditional of the state $R \subseteq X \times Y$ with respect to $X$ is given by `slicing' the relation
	\[ R|_X(x) = \{ y \in Y : (x,y) \in R \} \]
	which is nothing but $R$ itself.
\end{proof}

\subsection{Almost-sure Equality, Absolute Continuity and Supports}\label{sec:as_equality}

Conditionals in Markov categories are generally not uniquely determined, although there is still a sense in which they are unique. For example, the formula \eqref{eq:finstoch_conditional} only determines $\psi|_X$ on the set $\{ x : \pi(x) > 0 \}$, which we call the \emph{support} of $\psi_X$. Outside of the support, the conditional may be modified arbitrarily. Similarly, the formula for conditionals of Gaussian distributions \eqref{eq:conjugacy_formula} depends on a choice of generalized inverse, which is only unique on the appropriate support.

If the reader is familiar with measure-theoretic probability, they will recall that the essential uniqueness of conditionals, when they exist, is usually stated in terms of `almost-sure equality' and `absolute continuity'. These have established generalizations to Markov categories in general, as we now recall (Definitions~\ref{def:almost-sure} and~\ref{def:abscont} respectively). These definitions specialize to the measure-theoretic concepts by fixing a Markov category (Propositions~\ref{ex:almost_surely},~\ref{ex:borelstoch_ll}), but for a reader not familiar with measure-theoretic probability, the definitions can be taken as basic. An reference of measure theoretic terminology is given in \Cref{app:meas}.

\begin{definition}[{{\cite[5.1]{cho_jacobs}, \cite[13.1]{fritz}}}]\label{def:almost-sure}
	Let $\mu : I \to X$ be a distribution. Two morphisms $f, g : X \to Y$ are called \emph{$\mu$-almost surely equal} (written $f =_\mu g$) if \[ \langle \id_X, f \rangle\mu = \langle \id_X, g \rangle \mu \]
\end{definition}

\begin{proposition}\label{ex:almost_surely}
	For our example categories, the abstract definition of almost sure equality recovers the familiar meaning:
	\begin{enumerate}
		\item In $\finstoch$, $f, g : X \to D(Y)$ are $\mu$-almost surely equal iff the distributions $f(x) = g(x)$ agree for all $x$ with $\mu(x) > 0$ 
		\item In $\borelstoch$, $f, g : X \to \G(Y)$ are $\mu$-almost surely equal iff $f(x) = g(x)$ as measures for $\mu$-almost all $x$, that is the set $\{ x : f(x) \neq g(x) \}$ has $\mu$-measure $0$.
		\item In $\gauss$, if $\mu : I \to m$ is a distribution with support $S$(in the sense of \Cref{sec:recap_gauss}), then $f,g : m \to n$ are $\mu$-almost surely equal iff $f \circ x=g \circ x$ for all elements $x \in S$, seen as deterministic states $x : 0 \to m$.
		\item In $\rel^+$, if $M \subseteq X$ and $R,S : X \to \mathcal P^+(Y)$ are two left-total relations, then $R =_M S$ iff $R(x) = S(x)$ for all $x \in M$.
	\end{enumerate}
\end{proposition}
\begin{proof}
	The results for $\finstoch$ and $\borelstoch$ are given in \citep[13.2]{fritz} and \citep[3.19]{representable_markov}. The result for $\gauss$ is a strengthening of the result for $\borelstoch$. The morphisms $f, g : m \to n$ can be faithfully considered $\borelstoch$ maps $f, g : \R^m \to \G(\R^n)$, so we have $f(x) = g(x)$ for $\mu$-almost all $x$. Because $f,g$ are furthermore continuous functions and $\mu$ is equivalent to the Lebesgue measure on the support $S$, the equality almost everywhere can be strengthened to equality on all of $S$. 
\end{proof}

It follows directly from the definitions that conditional distributions are almost surely unique:

\begin{proposition}\label{prop:cond_as_unique}
	If $\psi : I \to X \otimes Y$ is a distribution and $\psi|_X, \psi|_X'$ are two morphisms satisfying \eqref{eq:cond_dist}, then 
	\begin{equation} \psi|_X =_{\psi_X} \psi|_X' \label{eq:conditionals_as_unique} \end{equation}
	That is, conditional distributions are unique almost surely with respect to the marginal $\psi_X$
\end{proposition}

 The important notion of \emph{absolute continuity} can now be formulated naturally in terms of almost-sure equality:

\begin{definition}[{{\cite[2.8]{representable_markov}}}]\label{def:abscont}
	Given two distributions $\mu,\nu : I \to X$, we say that $\mu$ is \emph{absolutely continuous} with respect to $\nu$, written $\mu \ll \nu$, if for all $f, g : X \to Y$ we have
	\[ f =_\nu g \text{ implies } f =_\mu g \]
\end{definition}

Absolute continuity lets us strengthen statements about almost-sure equality to actual equality. 

\begin{lemma}\label{lemma:suppunique}
	If $\mu : I \to X$, $f =_\mu g$ and $x \ll \mu$ then $fx = gx$
\end{lemma}

On the other hand, $\ll$ recovers the usual notion of absolute continuity in our example categories.

\begin{proposition}\label{ex:borelstoch_ll}
	\begin{enumerate}
		\item In $\finstoch$, $\mu \ll \nu$ iff $\nu(x) = 0$ implies $\mu(x) = 0$ 
		\item In $\borelstoch$, $\mu \ll \nu$ iff for all measurable sets $A$, $\nu(A) = 0$ implies $\mu(A) = 0$
		\item In $\gauss$, $\mu \ll \nu$ iff $\supp(\mu) \subseteq \supp(\nu)$ (in the sense of \Cref{sec:recap_gauss}).
		\item In $\rel^+$, $R \ll S$ iff $R \subseteq S$.
	\end{enumerate}
\end{proposition}
\begin{proof}
	The claim for $\finstoch$ follows immediately from \Cref{ex:almost_surely}, and the result for $\borelstoch$ is given in \cite[2.9]{representable_markov}. 
	\Cref{ex:almost_surely} implies that the support condition for $\gauss$ is sufficient. To see that it is also necessary, let $x \in \supp(\mu) \setminus \supp(\nu)$. Then we can find two affine functions $f,g$ which agree on $\supp(\nu)$ but $f(x) \neq g(x)$. Now $f=_\nu g$ but not $f=_\mu g$, hence $\mu \not \ll \nu$.
\end{proof}

In $\finstoch$, we can also rephrase $\mu \ll \nu$ as $\supp(\mu) \subseteq \supp(\nu)$ if we define $\supp(\mu) = \{ x : \mu(x) > 0 \}$. For the purposes of our development, it will suffice to consider the special case of the absolute continuity relation restricted to deterministic states and distributions. We take this as the categorical \emph{definition} of supports:

\begin{definition}\label{def:ll_support}
	If $x : I \to X$ is a \emph{deterministic state}, we say that \emph{$x$ lies in the support of $\mu$} if $x \ll \mu$.
\end{definition}

We obtain the following characterization

\begin{proposition}\label{prop:allsupports}
	\begin{enumerate}
		\item In $\finstoch$, $x \ll \mu$ iff $\mu(x) > 0$
		\item In $\borelstoch$, $x \ll \mu$ iff $\mu(\{x\}) > 0$
		\item In $\gauss$, $x \ll \mu$ iff $x \in \supp(\mu)$
		\item In $\rel^+$, $x \ll R$ if $x \in R$
	\end{enumerate}
\end{proposition}

It is crucial that the support of a distribution can change with the surrounding Markov category:
\begin{example}\label{ex:borelstoch_no_supp}
	Let $\mu = \mathcal N(0,1)$ be the standard normal distribution. When considered in $\gauss$, its support is $\R$ and in particular for all $x_0 \in \R$ we have $x_0 \ll \mu$. In $\borelstoch$, we have $x_0 \not \ll \mu$ because $\mu(\{x_0\}) = 0$.
\end{example}

This means that smaller Markov categories like $\gauss$ have a stronger notion of support, which in turn allows more interesting conditions to be evaluated. This is reminiscent of the tradeoff between expressiveness and well-behavedness discussed under the notion of ``well-behaved disintegrations'' in \cite{shan_ramsey}. 

By combining the notions of conditionals and support, we can now present an abstract theory of inference problems.

\subsection{Abstract Inference Problems}\label{sec:inferenceproblems}

Let $\C$ be a Markov category with all conditionals. In order to describe statistical inference categorically, we introduce the following terminology: 
\begin{enumerate}
	\item An \emph{observation} is a constant piece of data, that is a \emph{deterministic state} $o : I \to K$.
	\item An \emph{inference problem} over $X$ is a tuple $(K,\psi,o)$ of an object $K$, a joint distribution $\psi : I \to X \otimes K$ called the model and an observation $o : I \to K$.
\end{enumerate} 
The problem is then to infer the posterior distribution over $X$ conditioned on the observation $o$. An inference problem can either succeed, or fail if the observation $o$ is inconsistent with the model.

\begin{enumerate}
	\item  We say $(K,\psi,o)$ \emph{succeeds} if the observation lies in the support of the model, i.e. $o \ll \psi_K$. In that case, a \emph{solution} to the inference problem is the composite $\psi|_K \circ o : I \to X$ where $\psi|_K : K \to X$ is a conditional to $\psi$ with respect to $K$. The solution is also referred to as a \emph{posterior} for the problem.
	\item If $o \not\ll \psi_K$, we say that the inference problem \emph{fails} or is infeasible.
\end{enumerate}

\begin{proposition}\label{prop:inference_unique}
	Solutions to inference problems are unique, i.e. if $(K,\psi,o)$ succeeds and $\psi|_K, \psi|_K'$ are two conditionals then $\psi|_K(o) = \psi|_K'(o)$.
\end{proposition}
\begin{proof}
	Combine \Cref{lemma:suppunique} and \Cref{prop:cond_as_unique}.
\end{proof}

\begin{definition}\label{def:inferenceobseq}
	We call two inference problems \emph{observationally equivalent} if they either both fail, or they both succeed with equal posteriors.
\end{definition}

For the rest of this section, we will rederive \Cref{ex:og_gaussian_diagonal} in terms of the categorical machinery and show that it matches the conditioning procedure from \Cref{sec:gaussian_language}.

\begin{example}
	\label{ex:gaussian_diagonal}
	The \cref{ex:og_gaussian_diagonal} can be written as
	\begin{align*}
	X &\sim \mathcal N(0,1) \\
	Y &\sim \mathcal N(0,1) \\
	(X - Y) \,&\eqo\, 0
	\end{align*}
	which corresponds to the inference problem $(1,\mu,0)$ where $\mu : 0 \to 2 \otimes 1$ has covariance matrix
	\[ \Sigma = \begin{pmatrix} 1 & 0 & 1 \\ 0 & 1 & -1 \\ 1 & -1 & 2 \end{pmatrix} \]
	A conditional with respect to the third coordinate $Z$ is
	\[ \mu|_Z(z) = \begin{pmatrix}0.5 \\ 0.5\end{pmatrix}z + \mathcal N\begin{pmatrix}
	0.5 & 0.5 \\ 0.5 & 0.5
	\end{pmatrix} \]
	which can be verified by calculating \eqref{eq:cond_dist}. The marginal $\mu_Z = \mathcal N(2)$ is supported on all of $\R$, hence $0 \ll \mu_Z$ and by \Cref{prop:inference_unique} the composite
	\[ \mu|_{Z}(0) = \mathcal N\begin{pmatrix}
	0.5 & 0.5 \\ 0.5 & 0.5
	\end{pmatrix} \]
	is the uniquely defined solution to the inference problem.
\end{example}

The same inference problem would not have a solution when interpreted in $\borelstoch$ instead of $\gauss$. This is because $0 \not\ll \mu_Z$ (\Cref{ex:borelstoch_no_supp}). In $\borelstoch$, we can only condition on observations of positive probability; this agrees with the classical definition of conditional probability
\[ P(A|B) \defeq \frac{P(A \cap B)}{P(B)} \text{ if } P(B) > 0 \]
In $\gauss$, we can also condition on probability zero observations in a principled way because the notion of support is better behaved.


\section{Compositional Conditioning -- The Cond Construction}
\label{sec:cond}

In \Cref{sec:abstract_cond} we have seen that Markov categories with conditionals allow a general recipe for conditioning. In order to give \emph{compositional} semantics to a language with conditioning, we need to internalize the conditioning operation as a morphism.
The key step is to move from a closed inference problem $(K,\psi\colon I\to X\otimes K,o\colon I\to K)$, to 
\emph{open inference problems} or \emph{conditioning channels}, where $\psi$ is replaced with a morphism with more general domain, so that it can be composed.

With some care, we can turn these conditioning channels into a CD-category (\Cref{def:cd}, a monoidal category where every object has a comonoid structure). This allows us to give a denotational semantics to a CD calculus with conditioning, and in particular a denotational semantics for the Gaussian language with conditioning of \Cref{sec:gaussian_language}. Looking forward, in \Cref{sec:finstoch} we will show that this denotational semantics is fully abstract: it precisely captures the contextual equivalence from the operational semantics. 

The construction presented in this chapter is rather involved, but its well-definedness and properties are the central technical contribution of this article. From \Cref{sec:laws_cond} onwards, we can harness the good properties the construction enjoys to return to a higher level picture and use string diagrams for studying a graphical language of conditioning. \\

Let $\C$ be a Markov category, then a conditioning channel $X \obsto Y$ is given by a morphism $X \to Y \otimes K$ together with an observation (i.e. deterministic state) $o : I \to K$. This represents an intensional open program of the form
\begin{equation}
\label{eq:obs_normalform}
x : X \vdash \letin {(y,k) : Y \otimes K} {f(x)} {(k \eqo o); y}
\end{equation}
We think of $K$ as an additional hidden output wire, to which we attach the observation $o$. Such programs compose in the obvious way, by aggregating observations (\Cref{fig:obs_composition}). Two representations \eqref{eq:obs_normalform} are deemed equivalent if they contextually equivalent, that is roughly they compute the same posteriors in all contexts. \\

An important caveat is that the primary operation we formalize is that of an \emph{exact observation} $(\eqo o)$ where $o$ is a deterministic state. Binary \emph{exact conditioning} $(\eq)$ between two expressions may be encoded in terms of $(\eqo)$, for example as $(x==y)\eqo \true$ for finite sets or as $(x-y) \eqo 0$ for Gaussians. Generally, the choice of encoding does matter (\Cref{ex:borels_paradox}), so we consider this choice additional structure and focus on formalizing $(\eqo)$. \\

For modularity, we present the construction in two stages: In the first stage (\Cref{sec:obscat}) we form a category $\obs(\C)$ on the same objects as $\C$ consisting of the data \eqref{eq:obs_normalform} but without any quotienting. This adds, purely formally, for every observation $o : I \to X$ an observation effect $(\eqo o) : X \obsto I$.
In the second stage (\Cref{sec:condcat}) -- this is the core of the construction -- we relate these morphisms to the conditionals present in~$\C$, that is we quotient by contextual equivalence. The resulting quotient is called $\cond(\C)$. Under mild assumptions, this will have the good properties of a CD category, showing that conditioning stays commutative. We demonstrate the resulting reasoning methods in \Cref{sec:laws_cond} and \Cref{sec:ex-graphmod}. 

\subsection{Obs -- Open Programs with Observations}\label{sec:obscat}
For ease of notation, we will assume $\C$ is a strictly monoidal category, that is all associators and unitors are identities (this poses no restriction by \citep[10.17]{fritz}). We note that all constructions can instead be performed purely string-diagrammatically.
\begin{definition}
	\label{def:obs}
	The following data define a symmetric premonoidal category called $\obs(\C)$:
	\begin{itemize}
		\item the object part of $\obs(\C)$ is the same as $\C$
		\item morphisms $X \obsto Y$ are tuples $(K,f,o)$ where $K$ is an object of $\C$, $f \in \C(X,Y\otimes K)$ and $o \in \C_{\mathrm{det}}(I,K)$, representing \eqref{eq:obs_normalform}
		\item The identity on $X$ is $\Id_X = (I,\id_X,!)$ where $!=\id_I$.
		\item Composition is defined by
		\[ (K',f',o') \bullet (K,f,o) = (K' \otimes K, (f' \otimes \id_K)f, o' \otimes o). \]
		\item if $(K,f,o) : X \obsto Y$ and $(K',f',o') : X' \obsto Y'$, their (premonoidal) tensor product is defined as
		\[ (K' \otimes K, (\id_{Y'} \otimes \swap_{K',Y} \otimes \id_K)(f' \otimes f), o' \otimes o) \]
		\item There is an identity-on-objects functor $J : \C \to \obs(\C)$ that sends $f : X \to Y$ to $(I,f,!) : X \obsto Y$. This functor is strict premonoidal and its image central
		\item $\obs(\C)$ inherits symmetry and comonoid structure
	\end{itemize}
\end{definition}

Recall that a symmetric premonoidal category (due to \cite{power:premonoidal}) is like a symmetric monoidal category where the interchange law $(f_1 \otimes f_2) \circ (g_1 \otimes g_2) = f_1g_1 \otimes f_2g_2$ need not hold. This is the case because $\obs(\C)$ does not yet identify observations arriving in different order. This will be remedied automatically later when passing to the quotient $\cond(\C)$. For an observation $o : I \to K$, the conditioning effect $(\eqo o) : K \obsto I$ is defined by $(I,\id_K,o)$.

\paragraph{Notation:} 
A morphism $F : X \obsto Y$ in $\obs(\C)$ consists of a morphism $f : X \to Y \otimes K$ in $\C$ with an extra datum $o : I \to K$. It can be convenient to describe morphisms in $\obs(\C)$ by their underlying morphisms in $\C$, with extra `conditioning wires' into the object $K$ and observations attached to them. Informally, we will highlight these wires dashed and blue, simply to distinguish the codomain object $Y$ from the condition object $K$. Using this notation, composition and tensor in $\obs(\C)$ take the form shown in \Cref{fig:obs_composition}. In \Cref{sec:laws_cond}, we will use actual string diagram in $\cond(\C)$ as opposed to string diagrams in $\C$, which will obviate the need for blue wires. 

\begin{figure}[h]
	\centering
	\input{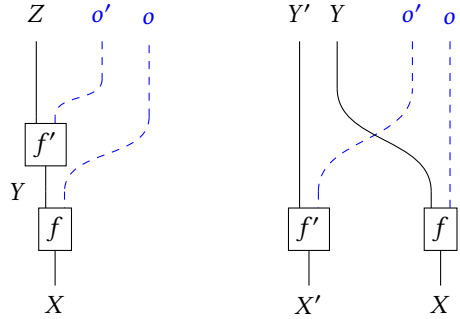}
	\caption{Composition and tensoring of morphisms in $\obs$}
	\label{fig:obs_composition}
\end{figure}

\subsection{Cond -- Contextual Equivalence of Inference Programs}
\label{sec:condcat}

Let us now assume that $\C$ has all conditionals. We wish to quotient $\obs$-morphisms by contextual equivalence, relating them to the conditionals which can be computed in $\C$. We know how to interpret \emph{closed} programs, because a state $(K,\psi,o) : I \obsto X$ is precisely an inference problem as in \Cref{sec:synth_conditioning}: If $o \not \ll \psi_K$, the observation does not lie in the support of the model and conditioning fails. If not, we form the conditional $\psi|_K$ in $\C$ and obtain a well-defined posterior $\mu|_K \circ o$. 

The observational equivalence (\Cref{def:inferenceobseq}) defines an equivalence relation on states $I \leadsto X$ in $\cond(\C)$. We will extend this relation to a congruence on arbitrary morphisms $X \leadsto Y$ by a general categorical construction. 

\begin{definition}
	\label{def:condeq}
	Given two states $I \leadsto X$ we define $(K,\psi,o) \sim (K',\psi',o')$ if they are observationally equivalent as inference problems, that is either
	\begin{enumerate}
		\item $o \ll \psi_K$ and $o' \ll \psi'_{K'}$ and $\psi|_K(o)= \psi'|_{K'}(o')$.
		\item $o \not \ll \psi_K$ and $o' \not \ll \psi'_{K'}$
	\end{enumerate}
\end{definition}

\begin{figure}[h]
	\vspace{-12pt}
	\centering
	\input{conditioning/obs_example.tikz}
	\caption{\Cref{ex:gaussian_diagonal} describes observationally equivalent states $0 \obsto 2$ in $\obs(\gauss)$}
	\label{fig:obs}
\end{figure}

We now give a general recipe to extend an equivalence relation on states to a congruence on arbitrary morphisms $f : X \to Y$. 

\begin{definition}
	\label{def:observational_quotient}
	Let $\mathbb X$ be a symmetric premonoidal category. An equivalence relation $\sim$ on states $\mathbb X(I,-)$ is called \emph{functorial} if $\psi \sim \psi'$ implies $f\psi \sim f\psi'$. We can extend such a relation to a congruence $\approx$ on all morphisms $X \to Y$ via
	\[ f \approx g \Leftrightarrow \forall A,\psi : I \to A \otimes X, (\id_A \otimes f)\psi \sim (\id_A \otimes g)\psi. \]
	The quotient category $\mathbb X/{\approx}$ is symmetric premonoidal. 
\end{definition}

We show now that under good assumptions, the quotient by conditioning (\Cref{def:condeq}) on $\mathbb X = \obs(\C)$ is functorial, and induces a quotient category $\cond(\C)$. The technical condition is that supports interact well with dataflow:
\begin{definition}\label{def:precise_supports}
	A Markov category $\C$ has \emph{precise supports} if the following are equivalent for all deterministic $x : I \to X$, $y : I \to Y$, and arbitrary $f : X \to Y$ and $\mu : I \to X$.
	\begin{enumerate}
		\item $x \otimes y \ll \langle \id_X, f\rangle \mu$
		\item $x \ll \mu$ and $y \ll fx$
	\end{enumerate}
	The word `support' here refers to \Cref{def:ll_support}.
\end{definition}

\begin{proposition}
	\label{prop:have_precise_supports}
	$\gauss$, $\finstoch$, $\borelstoch$ and $\rel^+$ have precise supports. 
\end{proposition}
\begin{proof}
This follows from the characterizations of $\ll$ in \Cref{prop:allsupports}. For $\gauss$, let $\mu$ have support $S$ and $f(x)=Ax+\mathcal N(b,\Sigma)$. Let $T$ be the support of $\mathcal N(b,\Sigma)$. The support of $\langle \id, f \rangle\mu$ is the image space $\{ (x,Ax+c) : x \in S, c \in T \}$. Hence $(x,y) \ll \langle \id, f \rangle\mu$ iff $x \ll \mu$ and $y \ll fx$. Similarly, for $\rel^+$, we readily verify
\[ x \in \{ (x,y) : x \in \mu, (x,y) \in f \} \Leftrightarrow x \in \mu \wedge y \in f(x) \]

For $\finstoch$, an outcome $(x,y)$ has positive probability under $\langle \id, f \rangle\mu$ iff $x$ has positive probability under $\mu$, and $y$ has positive probability under $f(-|x)$.

For $\borelstoch$, the measure $\psi = \langle \id, f \rangle \mu$ is given by
\[ \psi(A \times B) = \int_{x \in A} f(B|x) \mu(\mathrm dx) \]
Hence $\psi(\{(x_0,y_0)\}) = f(\{y_0\}|x)\mu(\{x\})$, which is positive exactly if $\mu(\{x_0\}) > 0$ and $f(\{y_0\}|x)>0$.
\end{proof}

\begin{theorem}
	\label{thm:functoriality}
	Let $\C$ be a Markov category that has conditionals and precise supports. Then $\sim$ is a functorial equivalence relation on $\obs(\C)$.
\end{theorem}
\begin{proof}
	Let $(K,\psi,o) \sim (K',\psi',o') : I \obsto X$ be equivalent states and $(H,f,v) : X \obsto Y$ be any morphism. We need to show that the composites
	\begin{equation}
	\label{eq:cond_fun_equiv}
	(H \otimes K, (f \otimes \id_K)\psi, v \otimes o) \sim (H \otimes K', (f \otimes \id_{K'})\psi', v \otimes o')
	\end{equation}
	are equivalent. We analyze different cases.
	
\paragraph{The states fail} If a state $(K,\psi,o)$ fails because $o \not \ll \psi_K$, then any composite must fail too. So both sides of \eqref{eq:cond_fun_equiv} fail and are thus equivalent. 

\paragraph{The composite fails} Assume from now that the states succeed and thus also have equal posteriors
\begin{equation} \psi|_{K}(o) = \psi'|_{K}(o') \label{eq:assumption_posteriors} \end{equation}

 We first show that the success conditions on both sides of \eqref{eq:cond_fun_equiv} are the same, so if the LHS fails so does the RHS. The ``precise supports'' axiom lets us split the success condition into two statements; that is the following are equivalent (and analogous for $\psi', o'$):
	\begin{enumerate}
		\item\label{it:joint_cond} $v \otimes o \ll (f_H \otimes \id_K)\psi$
		\item\label{it:separate_cond} $o \ll \psi_K$ and $v \ll f_H\psi|_K(o)$
	\end{enumerate}
  To see this, we instantiate \Cref{def:precise_supports} with the morphisms $\mu = \psi_K$ and $g=f_H \circ \psi|_K$, because the definition of the conditional $\psi|_K$ lets us recover
	\[ \langle g, \id_K \rangle \mu = (f_H \otimes \id_K)\psi. \]
	
	 It is clear that condition \ref{it:separate_cond} agrees for both sides of \eqref{eq:cond_fun_equiv}. Hence so does \ref{it:joint_cond}. 
	 
\paragraph{The composite succeeds} We are left with the case that both sides of \eqref{eq:cond_fun_equiv} succeed, and need to show that the composite posteriors agree
	\begin{equation} [(f \otimes \id_K)\psi]|_{H \otimes K}(v \otimes o) = [(f \otimes \id_{K'})\psi']|_{H \otimes K'}(v \otimes o') \label{eq:posterior_conclusion} \end{equation}
	We use a variant of the argument from \citep[11.11]{fritz} that double conditionals can be replaced by iterated conditionals. Consider the parameterized conditional 
	\[ \beta \defeq (f \circ \psi|_K)|_H : H \otimes K \to Y \]
	with univsonersal property
	\begin{equation} \input{conditioning/proof_beta_iter.tikz} \label{eq:beta_prop} \end{equation} 	
	Some string diagram manipulation shows that $\beta$ too has the universal property of the double conditional 
	\[ \beta = [(f \otimes \id_K)\psi]|_{H \otimes K} \]
	We check
	\begin{equation*} \input{conditioning/proof_beta_double.tikz} \end{equation*} 	
	which further reduces using \eqref{eq:beta_prop} to the desired	
	\begin{equation*} \input{conditioning/proof_beta_double_2.tikz} \end{equation*} 
	By specialization (\Cref{prop:specialization}), we can fix one observation $o$ in $\beta$ to obtain a conditional 
	\begin{equation}
	\beta(\id_H \otimes o) = (f \circ \psi|_K(o))|_H \label{eq:beta_special}
	\end{equation} 
	But this conditional agrees with $(f \circ \psi'|_K(o'))|_H$ by assumption \eqref{eq:assumption_posteriors}. Hence we can evaluate the joint posterior successively,
	\begin{align*}
	[(f \otimes \id_K)\psi]|_{H \otimes K}(v \otimes o) &= \beta(\id_H \otimes o) \circ v \\
	&\stackrel{\eqref{eq:beta_special}}= (f \circ \psi|_K(o))|_H \circ v \\
	&\stackrel{\eqref{eq:assumption_posteriors}}= (f \circ \psi'|_{K'}(o'))|_H \circ v \\
	&\stackrel{\text{symmetric}}= [(f \otimes \id_{K'})\psi']|_{H \otimes K'}(v \otimes o)
	\end{align*}
	establishing \eqref{eq:posterior_conclusion}.
\end{proof}

We can spell out the induced congruence $\approx$ on $\obs(X,Y)$ as follows:
\begin{proposition}
	\label{prop:simplified_cond}
	We have $(K,f,o) \approx (K',f',o') : X \obsto Y$ if and only if for all $\psi : I \to A \otimes X$, either
	\begin{enumerate}
		\item $o \ll f_K\psi_X$ and $o' \ll f'_{K'}\psi'_X$ and $[(\id_A \otimes f)\psi]|_K(o) = [(\id_A \otimes f')\psi']|_{K'}(o')$
		\item $o \not \ll f_K\psi_X$ and $o' \not \ll f'_{K'}\psi'_X$
	\end{enumerate}
	Furthermore, because $\C$ has conditionals, it is sufficient to check these conditions for $A=X$ and $\psi$ of the form $\cpy_X \circ \phi$. 
\end{proposition}

By \Cref{def:observational_quotient}, the quotient $\obs(\C)/{\approx}$ is a well-defined symmetric premonoidal category. We argue now that it is in fact monoidal. Checking the interchange means showing that the order of observations does not matter modulo $\approx$. We can derive this from a general statement about isomorphic conditions.

\begin{proposition}[Isomorphic conditions]
	\label{prop:isomorphic_conditions}
	Let $(K,f,o) : X \obsto Y$ and $\alpha : K \cong K'$ be an isomorphism. Then
	\[ (K,f,o) \approx (K',(\id_Y \otimes \alpha)f, \alpha o). \]
	In programming terms, the observations $(k \eqo o)$ and $(\alpha k \eqo \alpha o)$ are contextually equivalent. 
\end{proposition}
\begin{proof}
	Let $\psi : I \to A \otimes X$. We first notice that $o \ll \psi_K$ if and only if $\alpha o \ll \alpha \psi_K$, so the success conditions coincide. It is now straightforward to check the universal property
	\[ (\id_A \otimes f)\psi|_K = (\id_A \otimes ((\id_X \otimes \alpha)f))\psi|_{K'} \circ \alpha. \] 
	This requires the fact that isomorphisms are deterministic in a Markov category with conditionals \citep[11.26]{fritz}. The proof more generally works if $\alpha$ is deterministic and split monic.
\end{proof}

We can now give the Cond construction by means of quotienting $\obs(\C)$ modulo contextual equivalence.
\begin{definition} \label{def:condcat}
	Let $\C$ be a Markov category that has conditionals and precise supports. We define $\cond(\C)$ as the quotient category \[ \cond(\C) = \obs(\C)/{\approx} \] This quotient is a CD category, and the functor $J : \C \to \cond(\C)$ preserves CD structure.
\end{definition}

\subsection{Laws for Conditioning}\label{sec:laws_cond}

We will now establish convenient properties of $\cond(\C)$ in a purely abstract way. In terms of the internal language, those are the desired program equations for a language with exact conditioning. For example, the fact that $\cond(\C)$ is a well-defined CD category already implies that commutativity equation holds for such programs (\Cref{prop:all_lambdac_axioms}).\\

Secondly, we can draw string diagrams in the category $\cond(\C)$. These look like diagrams in $\C$ to which we add effects $(\eqo o) : X \to I$ for every observation $o : I \to X$. For example, \Cref{prop:isomorphic_conditions} states diagrammatically that for all isomorphisms $\alpha$ and observations $o$, we have
\[ \input{conditioning/isomorphic_conditions.tikz} \]

We begin by showing that passing to $\cond(\C)$ does generally not collapse morphisms which were distinct in $\C$, that is the functor $J$ is faithful for common Markov categories.

\begin{proposition}\label{prop:faithful}
Two morphisms $f, g : X \to Y$ are equated via $J(f) \approx J(g)$ if and only if
	\begin{equation} \forall \psi : I \to A \otimes X, (\id_A \otimes f)\psi = (\id_A \otimes g)\psi \label{eq:j_ident} \end{equation}
In particular, $J$ is faithful whenever $I$ is a separator. This is the case for $\gauss$, $\finstoch$, $\borelstoch$ and $\rel^+$.
\end{proposition}
\begin{proof}
Directly from the definition of $\approx$. 
\end{proof}

By construction, the states in $\cond(\C)$ are precisely inference problems up to observational equivalence. Any such problem either fails or computes a well-defined posterior, which gives rise to the following classification:

\begin{proposition}[States in $\cond$]
	\label{prop:closedterms} The states $I \obsto X$ in $\cond(\C)$ are of the following form:
\begin{enumerate}
\item There exists a unique failure state $\bot_X : I \obsto X$ given by the equivalence class of any $(K,\psi,o)$ with $o \not\ll \psi_K$.\footnote{it is a minor extra assumption that there exists a non-instance $o \not\ll \mu$ in $\C$; this should be the case in any Markov category of practical interest}
\item Any other state is equal to a conditioning-free posterior, namely $(K,\psi,o) \approx J(\psi|_K \circ o)$. That is diagrammatically 
	\[ \input{conditioning/cond_states.tikz} \]
\item Failure is ``strict'' in the sense that any composite or tensor with $\bot$ gives $\bot$. 
\item The only scalars $I \obsto I$ are $\id_I$ and $\bot_I$. Both are copyable, but $\bot_I$ is not discardable.
\end{enumerate}
\end{proposition}
\begin{proof}
By definition of $\sim$.
\end{proof}

\begin{corollary}\label{coro:success}
If $o \ll \psi$ then $(\psi \eqo o)$ succeeds without observable effect; in particular, because $o \ll o$, we can always eliminate tautological conditions
\[ \begin{tikzpicture}[scale=0.6]
	\begin{pgfonlayer}{nodelayer}
		\node [style=none] (18) at (0, 0) {=};
		\node [style=effect] (32) at (-1.5, 1) {$o$};
		\node [style=state] (35) at (-1.5, -1) {$o$};
		\node [style=none] (36) at (4, 0) {(empty diagram)};
	\end{pgfonlayer}
	\begin{pgfonlayer}{edgelayer}
		\draw (35) to (32);
	\end{pgfonlayer}
\end{tikzpicture}
 \]
\end{corollary}

The central law of conditioning states that after we enforce a condition, it will hold with exactness. In programming terms, this is the substitution principle \eqref{eqn:substlong}. Categorically, we are asking how the conditioning effect interacts with copying:

\begin{proposition}[Enforcing conditions]
	\label{prop:exactly}
	We have
	\[ (X,\cpy_X,o) \approx (X,o \otimes \id_X,o) \]
\end{proposition}
In programming notation, this is
\[ (x \eqo o); x \approx (x \eqo o); o \] 
and in string diagrams
\begin{equation} \input{conditioning/exactly.tikz} \label{eq:exact} \end{equation}
Note that the conditioning effect \emph{cannot} be eliminated; however after the condition takes place, the other wire can be assumed to now contain $o$. 
\begin{proof}
	Let $\psi : I \to A \otimes X$; the success condition reads $o \ll \psi_X$ both cases. Now let $o \ll \psi_X$ and let $\psi|_X$ be a conditional distribution for $\psi$. The following maps give the required conditionals 
	\begin{align*}
	[(\id_A \otimes \cpy_X)\psi]|_X = \langle \psi|_X, \id_X \rangle \quad 
	[(\id_A \otimes o \otimes \id_X)\psi]|_X = \psi|_X \otimes o
	\end{align*}
	as evidenced by the following string diagrams
	\[ \input{conditioning/proof_exact.tikz} \]
	Composing with $o$, we obtain the desired equal posteriors
	\[ \langle \psi|_X, \id_X \rangle o = \psi|_X(o) \otimes o = (\psi|_X \otimes o)(o) \]
	from determinism of $o$.
\end{proof}

\begin{corollary}[Initialization]
Conditioning a fresh variable on a feasible observation makes it assume that observation. Formally, if $o \ll \psi$ then 
\[ (\letin x \psi (x \eqo o); x) \approx o \]
\end{corollary}
\begin{proof}
Combining  \Cref{prop:exactly} and \Cref{coro:success}, we have
\[ \input{conditioning/initialization.tikz} \]
\end{proof}

\begin{corollary}[Idempotence]
Conditioning is idempotent, that is
\[ (x \eqo o); (x \eqo o) \approx (x \eqo o) \]
In other words, the conditioning effect is copyable (but not discardable).
\end{corollary}
\begin{proof}
Again by \Cref{prop:exactly} and \Cref{coro:success} we obtain
\[ \input{conditioning/idempotence.tikz} \]
\end{proof}

We note that this does not imply that \emph{every} effect in $\cond(\C)$ is copyable, only that exact observations are. 

\begin{proposition}[Aggregation]
Conditions can be aggregated
\[ \input{conditioning/aggregation.tikz} \]
\end{proposition}
\begin{proof}
By definition of the monoidal structure of $\obs$.
\end{proof}

\subsection{Example: Graphical Models and Conditioning}\label{sec:ex-graphmod}
We demonstrate the power of our conditioning laws by briefly revisiting graphical models as mentioned in the introduction of \Cref{sec:synth_conditioning}: Every graphical model can be turned into a string diagram, where the independence structure of the graphical model translates into a factorization of the diagram. For example, in the model \eqref{eq:graphicalmodel} of variables $X,Y$ which are conditionally independent on $W$, the joint distribution $\psi$ can be factored as follows
\[ \input{conditioning/cond_ind_state.tikz} \] 
Using the conditioning effects in $\cond(\C)$, we can now incorporate \emph{observed nodes} into this language.
\begin{equation}
\input{conditioning/cond_ind_observed.tikz} \label{eq:cond_ind_observed}
\end{equation}

We want to argue that once the `common cause' $W$ has been observed, $X$ and $Y$ become independent: We can show this purely using graphical reasoning: Applying repeatedly \Cref{prop:exactly}, idempotence of scalars and determinism of $w$, we obtain that \eqref{eq:cond_ind_observed} is the product of its marginals:
\[ \input{conditioning/cond_ind_proof.tikz} \]


\section{Denotational Semantics and Contextual Equivalence}\label{sec:finstoch}

In \Cref{sec:cond} we introduced the $\cond$ construction (\Cref{def:condcat}) as a way of building a category that accommodates the abstract inference for Markov categories (\Cref{sec:abstract_cond}). As we have seen, we can interpret the CD calculus (\Cref{sec:synth_conditioning}) in categories built from the $\cond$ construction, and this forms a probabilistic programming language with exact conditioning. In this final section, we will work out in detail what the $\cond$ construction does when applied to our specific example settings of finite and Gaussian probability.

In \Cref{sec:denotational}, we show that the Gaussian language (\Cref{sec:gaussian_language}) has fully abstract denotational semantics in $\cond(\gauss)$: equality in the category coincides with the operational contextual equivalence from~\Cref{sec:opsem}.

In \Cref{sec:finprojstoch}, we conduct the same analysis for finite probability and show that $\cond(\finstoch)$ consists of substochastic kernels up to \emph{automatic normalization}. In \Cref{sec:straightline}, we spell out the relationship between the admissibility of automatic normalization and the expressibility of branching in the language.

\subsection{Full Abstraction for the Gaussian Language}
\label{sec:denotational}

The Gaussian language embeds into the internal language of $\cond(\gauss)$, where $x \eq y$ is translated as $(x - y) \eqo 0$. A term $\vec x : R^m \vdash e : R^n$ denotes a conditioning channel $\sem{e} : m \obsto n$.

\begin{proposition}[Correctness]
	If $(e,\psi) \red (e',\psi')$ then $\sem{e}\psi = \sem{e'}\psi'$. If $(e,\psi) \red \bot$ then $\sem{e} = \bot$.
\end{proposition}
\begin{proof}
	We can faithfully interpret $\psi$ as a state in both $\gauss$ and $\cond(\gauss)$. If $x \vdash e$ and $(e,\psi) \red (e',\psi')$ then $e'$ has potentially allocated some fresh latent variables $x'$. We show that
	\begin{equation}
	\letin{x}{\psi} (x,\sem{e}) = \letin{(x,x')}{\psi'} (x,\sem{e'}).
	\end{equation}
	This notion is stable under reduction contexts. 
	
	Let $C$ be a reduction context. 
	Then
	\begin{align*}
	&\letin x \psi (x,\sem{C[e]}(x)) \\
	&= \letin x \psi \letin y {\sem{e}(x)} (x,\sem{C}(x,y)) \\
	&= \letin {(x,x')} {\psi'} \letin y {\sem{e'}(x,x')} (x,\sem{C}(x,y)) \\
	&= \letin {(x,x')} {\psi'} (x,\sem{C[e']})
	\end{align*}
	Now for the redexes
	\begin{enumerate}
		\item The rules for $\mathrm{let}$ follow from the general axioms of value substitution in the internal language
		\item For $\normal()$ we have $(\normal(), \psi) \red (x',\psi \otimes \mathcal N(0,1))$ and verify
		\begin{align*}
		&\letin x \psi (x,\sem{\normal()}) \\
		&= \psi \otimes \mathcal N(0,1) \\
		&= \letin {(x,x')} {\psi \otimes \mathcal N(0,1)} (x,\sem{x'})
		\end{align*}
		\item For conditioning, we have $(v\eq w,\psi) \red ((), \psi|_{v=w})$. We need to show
		\begin{align*}
		&\letin x \psi (x, \sem{v \eq w}) = \letin x {\psi|_{v=w}} (x,())
		\end{align*}
		Let $h=v-w$, then we need to the following morphisms are equivalent in $\cond(\gauss)$:
		\[ \input{conditioning/cond_exact.tikz} \]
		Applying \Cref{prop:closedterms} to the left-hand side requires us to compute the conditional $\langle \id, h\rangle\psi|_2 \circ 0$, which is exactly how $\psi|_{h=0}$ is defined.
	\end{enumerate}
	\vspace{-12pt}
\end{proof}

\begin{theorem}[Full abstraction]
	\label{prop:full_abstraction}
	$\sem{e_1} = \sem{e_2}$ if and only if $e_1 \approx e_2$ (where $\approx$ is contextual equivalence, \Cref{def:ctxequiv}).
\end{theorem}
\begin{proof}
	For $\Rightarrow$, let $K[-]$ be a closed context. Because $\sem{-}$ is compositional, we obtain $\sem{K[e_1]} = \sem{K[e_2]}$. By \Cref{prop:evaluation} If both succeed, we have reductions $(K[e_i], !) \red^* (v_i,\psi_i)$ and by correctness $v_1\psi_1 = \sem{K[e_1]} = \sem{K[e_2]} = v_2\psi_2$ as desired. If $\sem{K[e_1]} = \sem{K[e_2]} = \bot$ then both $(K[e_i], !) \red^* \bot$. 
	
	For $\Leftarrow$, we note that $\cond$ quotients by contextual equivalence, but all Gaussian contexts are definable in the language. 
\end{proof}

\subsection{Contextual Equivalence for Finite Probability}\label{sec:finprojstoch}

With programs over a finite domain, we can understand conditioning in terms of rejection sampling.
This means that we run a program $N$ times, with different random choices each time. We reject those runs that violate the conditions, and then we resample from the among the acceptable results. As $N\to \infty$, this random distribution converges to the probability distribution that the program describes.

The following reformulation is semantically equivalent: For a closed program, suppose the program would return some value $x_1$ with probability $p_1$, value $x_2$ with probability $p_2$, and so on.
Then the probability that the program will not fail is $Z=\sum_i p_i$. The result of rejection sampling is a program that actually returns $x_i$ with probability $\frac {p_i} {Z}$, so that we have a normalized probability distribution over $\{x_1\dots x_n\}$, i.e.~$\sum \frac {p_i}{Z}=1$. The quantity $Z$ is called the \emph{normalization constant} of the program, or sometimes \emph{model evidence}.

For example, the program
\[ \letin {x} {\bernoulli (0.4)} {\letin {y} {\bernoulli (0.4)} {x \eq y; x}}\]
will fail the condition with probability $2\cdot 0.4\cdot 0.6 = 0.48$, return true with probability $0.4^2=0.16$, and return false with probability $0.6^2=0.36$.
Under rejection sampling, once we renormalize, the program is equivalent to $\bernoulli(\frac {0.16}{0.36})\approx\bernoulli(0.44)$. 

Rejection sampling makes sense for closed programs. For programs with free variables, we can still understand a program that rejects runs that violate the conditions, but normalization is more subtle. For example, in the program
\begin{equation}\label{eqn:finprobintroB}{\letin {y} {\bernoulli (0.4)} {x \eq y; x}}\end{equation}
the normalizing constant is either~$0.4$ or~$0.6$ depending on the value of~$x$.
If we normalize regardless of the value of $x$, then the meaning of the program must change, because it would simply return $x$,  and the context 
\[ \letin {x} {\bernoulli (0.4)} {[-]}\]
distinguishes this.

There is nonetheless some normalization that can be done in straight-line programs, since e.g. the meaning is \emph{not} changed by prefixing a program with a closed program. It is for example safe to regard program~\eqref{eqn:finprobintroB} as equivalent to
\begin{equation} 
\letin {z} {\bernoulli (0.2)} {z\eq \false ; \letin {y} {\bernoulli (0.4)} {x \eq y; x}} \label{eq:autonorm_example}
\end{equation}
because the difference in normalizing constant will be the same for both values of~$x$.

Semantically, the interpretation of a program with free variables is a stochastic kernel, and one involving rejection too is a substochastic kernel (\Cref{def:substoch}). As we show, we can accommodate multiplication by a constant if it is uniform across all arguments; this is what we call `projectivized' substochastic kernels, by analogy with the construction of a projective space from a vector space. 

\begin{definition}
  The CD-category $\finprojstoch$ of \emph{projectivized substochastic kernels} is a quotient of the CD-category of $\finsubstoch$ of substochastic maps: 
	\begin{enumerate}
		\item objects are finite sets $X$
		\item morphisms $X \to Y$ are equivalence classes $[p]$ of substochastic kernels $p(y|x)$ up to a scalar. That is we identify $p$ and $q$ if there exists a number $\lambda > 0$ such that $p(y|x) = \lambda \cdot q(y|x)$ for all $x\in X, y \in Y$. In this circumstance we write $p\propto q$.
	\end{enumerate}
      \end{definition}
      It is routine to verify that the monoidal and CD category structure are preserved by this quotient.
      
\begin{theorem}\label{thm:condfinstoch}
The CD-categories	$\cond(\finstoch)$ and $\finprojstoch$ are equivalent.
\end{theorem}
\begin{proof}
Sketch. Given a conditioning channel $Q : X \obsto Y$ presented by a finite set of observations $K$, a probability kernel $q(y,k|x)$ and an observation $k_0 \in K$, we associate to it the subprobability kernel $\rho_Q : X \to D_{\leq 1}(Y)$ given by the likelihood function
\[ \rho_Q (y|x) = q(y,k_0|x) \]
Conversely, we associate to every subprobability kernel $\rho : X \to D_{\leq 1}(Y)$ a conditioning channel $Q_\rho = (\{0,1\}, q_\rho, 1)$ with a single boolean observation $b \eqo 1$, defined as
\[ q_\rho(y,b|x) = b \cdot \rho(y|x) + (1-b)\cdot (1-\rho(y|x)). \]
We recover the subprobability kernel $\rho_Q$ from the conditioning channel $Q$ that way: Given any distribution $p(x)$, the posterior in \Cref{prop:simplified_cond} is given by 
\[ \frac{p(x)q(y,k_0|x)}{\sum_{x,y} p(x)q(y,k_0|x)} \]
We see that $q_\rho$ computes the same posterior, namely 
\[ \frac{p(x) q_\rho(y,1|x)}{\sum_{x,y} p(x) q_\rho(y,1|x)} = \frac{p(x) \rho(y|x)}{\sum_{x,y} p(x) \rho(y|x)} = \frac{p(x) q(y,k_0|x)}{\sum_{x,y} p(x) \rho(y,k_0|x)}\]

On the other hand, the subprobability kernel $\rho$ can be recovered from $Q_\rho$ up to a constant. For a uniform prior $p(x) = 1/|X|$, the posterior under $Q_\rho$ in \Cref{prop:simplified_cond} becomes
\[ \frac{\rho(y|x)}{\sum_{x,y} \rho(x,y)} \]
from which we can read off $\rho(y|x)$ up to the constant in the denominator.
\end{proof}

Identifying scalar multiples is necessary because the Cond construction by definition `normalizes automatically'. That is, it considers two conditioning channels equivalent if they compute the same posterior distributions for all priors. We will explore the relationship with model evidence and branching in \Cref{sec:straightline}. 

We briefly showcase some of the structure of this category, by characterizing the discardable morphisms, and observing that conditioning gives a commutative monoid structure. The latter gives a characterization for the finite uniform distributions as units for conditioning.
\begin{example}
	A projectivized subprobability kernel $p : X \to Y$ is discardable (\Cref{def:copydisc}) if and only if there exists a constant $\lambda \neq 0$ such that
	\[ \forall x, \sum_y p(y|x) = \lambda \]
	As an instance of \Cref{prop:closedterms}, in particular, to give a state in $\finprojstoch(1,Y)$ is to give either a normalizable distribution $p \in \finstoch(1,Y)$ or the failure kernel $\bot_Y = 0$.
\end{example}

\begin{definition}[Conditioning product]
In $\cond(\finstoch)$, we define an exact conditioning operation $x \eq y$ by exactly observing $\true$ from the boolean equality test $(x == y)$. We define a morphism $\bullet : X \times X \obsto X$ by
\[ x \bullet y \defeq (x \eq y); x \]
In terms of projectivized subprobability kernels, this is
\[ \bullet(z|x,y) = \begin{cases}
1 & z = x = y \\
0 & \text{otherwise}
\end{cases} \]
We call $\bullet$ the \emph{conditioning product}. 
\end{definition}
Concretely for subdistributions $p(x), q(x)$, the subdistribution $(p \bullet q)$ has the product of mass functions
\[ (p \bullet q)(x) = p(x) \cdot q(x) \]

\begin{proposition}\label{prop:uniform_unit}
The conditioning product defines a commutative monoid structure on $X$, where the unit is given by the uniform distribution $u_X : 1 \obsto X$.
\end{proposition}
\begin{proof}
The operation $\bullet$ is commutative and associative already in $\finsubstoch$, however it does not have a unit. Conditioning with the uniform distribution produces a global factor of $1/|X|$, which is cancelled by the proportionality relation. Therefore, $u_X$ is a unit for $\bullet$ in $\cond(\finstoch)$.
\end{proof}

This is intuitive in programming terms: Observing from a uniform distribution gives no new information. Such a conditioning statement can thus be discarded. 

\paragraph{Aside on non-determinism}
We show the analogous version of \Cref{thm:condfinstoch} for nondeterminism. The Cond construction here does nothing more than add the possibility for failure (zero outputs) in a systematic way. 
\begin{proposition}
$\cond(\rel^+) \cong \rel$.
\end{proposition}
\begin{proof}
Given a conditioning channel $(K,R,k_0)$ with $R \subseteq X \times Y \times K$ left-total in $X$, we define a possibly non-total relation $R' \subseteq X \times Y$ by $R' = \{ (x,y) : (x,y,k_0) \in R \}$. On the other hand, given $R'$, we form the conditioning channel $(2,R'',1)$ with left-total relation $R'' \subseteq X \times Y \times 2$ defined as
\[ (x,y,b) \in R'' \stackrel{\text{def}}\Leftrightarrow ((x,y) \in R' \Leftrightarrow (b=1)) \]
These constructions are easily seen to be inverses. Any relation $R'$ is recovered from $R''$ and we have $(K,R,k_0) \approx (2,R'',1)$ because \Cref{prop:simplified_cond} boils down to checking that $(x,y,k_0) \in R \Leftrightarrow (x,y,1) \in R''$.
\end{proof}
The conditioning product $\bullet$ in $\rel$ is the relation $\{ (x,x,x) : x \in X \}$, and on states we have $R \bullet S = R \cap S$. The conditioning product has a unit $v_X : 1 \to X$ given by the maximal subset $v_X = X$.

\subsection{Automatic Normalization and Straight-line Inference}\label{sec:straightline}

By \emph{automatic normalization} we mean that two (open) probabilistic programs which differ by an overall normalization constant $Z$ are considered equivalent, as formalized in \Cref{sec:finprojstoch}. As a consequence, the precise value of the normalization constant cannot be extracted operationally from such programs, which is a limitation whenever $Z$ is itself a quantity of interest. On the other hand, auto-normalization is a convenient optimization, as seen in \eqref{eq:autonorm_example} or \Cref{prop:uniform_unit}. \\

In this section, we argue that validity of auto-normalization is tied to the form of branching available in language under consideration. We distinguish \emph{straight-line inference programs} with a static structure of conditions from programs where we can dynamically choose whether to execute conditions or not. This is sufficient for inference in Bayesian networks, in which the observed nodes are determined statically. For example, the Bayesian network depicted in \eqref{eq:graphicalmodel} corresponds to the straight-line program in \Cref{fig:cdcalc}. Auto-normalization is valid for straight-line programs but not for those with more general branching. The quotient from \Cref{sec:finprojstoch} arises naturally from studying contextual equivalence of a straight-line inference language. In semantical terms, this means $\cond(\finstoch)$ does not have coproducts. 

\medskip

We consider the CD-calculus (\Cref{sec:cd-calc}) with base types $X$ for all finite sets and function symbols $\ct f$ for all subprobability kernels $f : X \to D_{\leq 1}(Y)$:
\begin{align*}
t ::= x \s () \s (t,t) \s \pi_it \s \ct f(t) \s \letin x {t_1} t_2 
\end{align*} 
This language can express scoring and exact conditioning because there exist suitable subprobability kernels 
\[ \mathsf{score}_p \in D_{\leq 1}(1) \text{ and } (\eq) : X \times X \to D_{\leq 1}(1) \]
We denote the this calculus $\ppisl$, for \emph{straight-line inference}. Following the development in Section~\ref{cd:structural}, the language $\ppisl$ has canonical denotational semantics in $\finsubstoch$. (In fact, $\ppisl$ is precisely the internal language of $\finsubstoch$ as a CD category.)

We also consider a richer language, $\ppi$, which contains the syntax of $\ppisl$ and also if-then-else branching:
\[t ::= \dots \s \ite {t_1}{t_2}{t_3}\]
with typing rule
\[ \infer{\Gamma \vdash \ite{t_1}{t_2}{t_3} :A }{\Gamma \vdash t_1 : 2 \quad \Gamma \vdash t_2 : A \quad \Gamma \vdash t_3 : A} \]
This language can also be interpreted in $\finsubstoch$, via
\[ \sem{\ite{t_1}{t_2}{t_3}}(a|\gamma) = \sem{t_2}(a|\gamma) \cdot \sem{t_1}(\true|\gamma) + \sem{t_3}(a|\gamma) \cdot \sem{t_1}(\false|\gamma) \]
Categorically, this makes use of the distributive coproducts in $\finsubstoch$ \cite{power:distributive_freyd,staton:commutative_semantics}. $\ppi$ is a commonly considered probabilistic language, and subprobability kernel semantics are already known to be fully abstract, though we rederive this here.

As explained in the introduction to this section, a program in the language $\ppi$ or $\ppisl$ is typically executed by some sort of inference engine, which tries to sample (usually approximately) from the posterior distribution it defines. In the semantics, we express this top-level normalization for a finite set~$X$ using the function $\norm : \subd(X) \to \subd(X)$ which is defined as
\[ \norm(\varphi)(x) = \begin{cases}
    \tfrac 1 Z \cdot \varphi(x) & \text{where } Z = \sum_{x\in X} \varphi(x)\neq 0
  \\ 0& \text{where }\forall x.\,\varphi(x)=0\end{cases}
\]
The zero distribution is mapped to itself, signaling failure of normalization. 

\begin{definition}
Two closed $\ppi$ programs $s,t:X$ are called \emph{observationally equivalent}, written $s \approx t$, if the normalized distributions they define are equal, that is $\norm(\sem{s}) = \norm(\sem{t})$.
\end{definition}

We say that two open programs $\Gamma \vdash s, t:X$ are contextually equivalent if under every closed context $C[-]$ we have $C[s] \approx C[t]$. The distinguishing power crucially depends on the fragment of the language we are allowed to use in the contexts $C[-]$.

\begin{definition}
The terms $s,t$ are called \emph{straight-line equivalent}, written $s \approx_{\ppisl} t$, if for every closed context $C[-]$ in $\ppisl$ (without branching), we have $C[s] \approx C[t]$. The terms $s,t$ are called \emph{branching equivalent}, written $s \approx_{\ppi} t$, if for every closed context $C[-]$ in $\ppi$ (possibly involving branching), we have $C[s] \approx C[t]$.
\end{definition}

The following proposition shows that straight-line equivalence can distinguish subprobability kernels up to a constant. 
\begin{proposition}
Two open programs are straight-line equivalent iff their denotations are proportional.
\[ s \approx_{\ppisl} t \Leftrightarrow \sem{s} \propto \sem{t} \]
That is $\finprojstoch$ is fully abstract for the language $\ppisl$.
\end{proposition}
\begin{proof}
$\Leftarrow$ Is is easy to show that the semantics of all $\ppisl$ constructs are linear or bilinear and hence respect the relation $\propto$. $\Rightarrow$ Let $s \approx_\ppisl t$ and consider the straight-line context
\[ C[t] \defeq \letin x {u_X} (x,t) \]
where again $u_X$ denotes the uniform distribution on the finite set $X$. Its denotation is
\[ \sem{C[t]}(x,y) = \frac{1}{|X|} \sem{t}(y|x) \]
By assumption $\sem{C[s]} \propto \sem{C[t]}$, so we have $\sem{s} \propto \sem{t}$. 
\end{proof}

This gives a clear interpretation of \Cref{thm:condfinstoch}. In our current terminology, the Cond construction aims to give a canonical semantics for straight-line inference programs: the construction presents a normal form for straight-line programs up to straight-line equivalence. The lack of branching is reflected in the fact that unlike $\finsubstoch$, $\finprojstoch$ does not have coproducts. 

\paragraph{Branching inference}
Up to straight-line equivalence, programs which differ by a global constant are not distinguishable, that is auto-normalization is valid. This changes if we allow branching in the contexts, because branching can be used to extract the normalization constant. This trick is fundamental to so-called `Bayesian model selection'. If $y$ is a closed program, then the boolean program
\[ \ite{\mathsf{bernoulli(0.5)}} {y;\true} {\false} \]
normalizes to a distribution which returns $\true$ with probability
\[ p=\frac {0.5\cdot Z}{0.5\cdot Z + 0.5}=\frac{Z}{Z+1} \qquad \text{ where } Z = \sum_{x} \sem{y}(x) \]
and $\false$ with probability $\frac 1 {Z+1}$.
Because the assignment $Z \mapsto Z/(Z+1)$ is a bijection $[0,\infty) \to [0,1)$, we can recover $Z$ from the probability $p$. It follows that $\finsubstoch$ is fully abstract for $\ppi$. 

\begin{proposition}
Two open programs $s,t$ are branching equivalent if their denotations are equal as subprobability kernels.
\[ s \approx_\ppi t \Leftrightarrow \sem s = \sem t:\sem\Gamma\to D_{\leq 1}(\sem X) \]
That is $\finsubstoch$ is fully abstract for the language $\ppi$.
\end{proposition}
\begin{proof}
From straight-line equivalence, we know that $\sem{s} = \lambda \cdot \sem{t}$. We then use the `Bayesian model selection' trick to show $\lambda = 1$.
\end{proof}

The tradeoff between straight-line inference and branching inference is an interesting design decision: Branching inference is more general and allows us to extract the normalization constant. On the other hand, restricting ourselves to straight-line inference, we are free to normalize at any point, which leads to an appealing equational theory. For example, an `uninformative' observation form a uniform distribution can be eliminated (\Cref{prop:uniform_unit}).

We emphasize that the crucial difference between $\ppisl$ and $\ppi$ lies in putting conditions in branches. Even in $\ppisl$, we can still implement if-then-else when the branches do not involve conditions, because we can make use of the probability kernel
\[
\mathsf{ite} : 2 \times X \times X \to D(X)
\]
with $\mathsf{ite}(1,x,y)=\delta_x$ and $\mathsf{ite}(0,x,y)=\delta_y$. If $t_1,t_2$ are discardable (condition-free) terms, we can define \[ (\ite c {t_1} {t_2}) \defeq \mathsf{ite}(c,t_1,t_2). \]
In $\ppisl$, the structure of conditions is static, while it is dynamic in $\ppi$. 


\section{Context, Related Work and Outlook}\label{sec:cond_outlook}

\subsection{Symbolic Disintegration, Consistency and Paradoxes}\label{sec:outlook_paradoxes}
\label{sec:paradoxes}
Our line of work can be regarded as a synthetic and axiomatic counterpart of the symbolic disintegration of \cite{shan_ramsey} (see also~\citep{psi,delayed_sampling,shan:disint,shan}).
That work provides in particular verified program transformations to convert an arbitrary probabilistic program of type $\rv \otimes \tau$ 
to an equivalent one that is of the form 
\[ \letin x {\mathrm{lebesgue}()} {\letin y M {(x,y)}} \]
Now the exact conditioning $x\eqo o$ can be carried out by substituting $o$ for $x$ in $M$. 
We emphasize the similarity to our treatment of \emph{inference problems} in \Cref{sec:synth_conditioning}, as well as the role that coordinate transformations play in both our work \cite{2021compositional} and \citep{shan_ramsey}. One language novelty in our work is that exact conditioning is a first-class construct in our language, as opposed to a whole-program transformation, which makes the consistency of exact conditioning more apparent. \\

Consistency is a fundamental concern for exact conditioning. \emph{Borel's paradox} is an example of an inconsistency that arises if one is careless with exact conditioning (\citep[Ch.~15]{jaynes}, \citep[\S3.3]{jacobs:paradoxes}): It arises when naively substituting equivalent equations within $(\eq)$. For example, the equation $x - y = 0$ is equivalent to $x/y = 1$ over the (nonzero) real numbers. Yet, in a hypothetical extension of our language which allows division, the following programs would not contextually equivalent, as discussed in \Cref{ex:borels_paradox}:
\begin{equation*}
\begin{array}{l}
\mlstinline{x = normal(0,1)} \\
\mlstinline{y = normal(0,1)} \\
\mlstinline{x-y =:= 0} \\
\end{array}
\not \equiv
\begin{array}{l}
\mlstinline{x = normal(0,1)} \\
\mlstinline{y = normal(0,1)} \\
\mlstinline{x/y =:= 1}
\end{array}
\end{equation*}

For that reason, we make it clear in our treatment of inference problems (\Cref{sec:inferenceproblems}) that conditioning on a deterministic observation $(\eqo)$ is the fundamental notion. Binary conditioning $(\eq)$ is a derived notion which involves further choices, and those choices are not equivalent. Our approach also makes it clear that we should always condition on random variables directly, and not on (boolean) predicates: By presenting conditioning as an algebraic effect, the expressions $(s \eq t) : \unit$ and $(s == t) : \mathrm{bool}$ have a different formal status and can no longer be confused. 

\subsection{Contextual Equivalence for Exact Conditioning languages}

In this article, we have emphasized the role of program equations for manipulating probabilistic programs, and based the Cond construction on an analysis of contextual equivalence of straight-line inference (\Cref{sec:finstoch}).

While we have fully characterized the case of finite probability (\Cref{sec:finprojstoch}), a corresponding explict characterization of contextual equivalence for the Gaussian language is still outstanding. We have given partial results in that direction (see \cite{2021compositional}) in the form of a sound equational theory for contextual equivalence. The classification of effects $n \to 0$ in $\cond(\gauss)$ is not straightforward: it is not true that every effect is observing from a unique distribution $0 \to n$, as for example $(\eq) : 2 \to 0$ is not of that form. We believe that by passing to an extension category of Gaussians $\gauss \to \mathsf{GaussEx}$, we can obtain the desired duality and achieve a more explicit characterization. 

It is a further challenge to find semantics for exact conditioning with branching. Automatic normalization is no longer valid here (\Cref{sec:straightline}) and the subtleties of \cite{jacobs:paradoxes} have to be accounted for. Mathematically, this would be an extension of the Cond construction which produces a distributive Freyd category \cite{power:distributive_freyd}. An example of a categorical model of the Beta-Bernoulli process with branching (but no first-class conditioning) is in \cite{staton:betabernoulli}.

\subsection{Other Directions}

\paragraph{Categorical tools}
Once a foundation is in algebraic or categorical form, it is easy to make connections to and draw inspiration from a variety of other work: The $\obs$ construction (\Cref{def:obs}) that we considered here is reminiscent of lenses~\citep{lenses} and the Oles construction~\citep{hermida-tennent}. These have recently been applied to probability theory~\citep{smithe:bayesian}, quantum theory~\citep{huotstatonqpl} and reversible computing~\citep{heunen-kaarsgaard}. The details and intuitions are different, but a deeper connection or generalization may be profitable in the future. 

\paragraph{Probabilistic logic programming} 
The concept of exact conditioning is reminiscent of unification in Prolog-style logic programming. Our presentation in \cite{2021compositional} is partly inspired by the algebraic presentation of predicate logic of~\cite{staton:predicate_logic}, which has a similar signature and axioms. Logic programming is also closely related to relational programming, and we note that our laws for conditioning are reminiscent of graphical presentations of categories of linear relations \citep{baez:control,bsz,pbsz}.

\textsc{ProbLog} \citep{problog} supports both logic variables as well as random variables within a common formalism. We have not considered logic variables in conjunction with the Gaussian language, but a challenge for future work is to bring the ideas of exact conditioning closer to the ideas of unification, both practically and in terms of the semantics. This is again related to the extension $\mathsf{GaussEx}$ by ``improper priors'', which are a unit for the conditioning product in the same way uniform distributions are in finite probability (\Cref{prop:uniform_unit}). The connections with logic programming are spelled out in more detail in \cite[Section~20.2]{dariothesis}.

\paragraph{Implementation} The purpose of our Gaussian language was to give a minimalistic calculus in which to study the novel effect of conditioning in isolation. The close fit of the denotational semantics to the language was thus expected, and can be seen as an instance of letting semantics inspire language design. To extend our calculus to a full-blown programming language, one can make use of the general framework of algebraic effects to combine conditioning with other effects like memory or recursion. For example, we can treat higher-order functions by modelling the language on a presheaf category, which is cartesian closed. The operational semantics easily extend to a full language, for which we have given implementations in Python and F\# \citep{gaussianinfer}. 

\begin{acks}
One starting point for this work was a question from Ohad Kammar, and we thank him for that. It has been helpful to discuss developments in this work with many people over the years, and we are also grateful to our anonymous reviewers for helpful suggestions. This material is partly based on work supported by an Royal Society University Research Fellowship, ERC Project BLAST, and AFOSR Award No.~FA9550–21–1–003. 
\end{acks}

\bibliography{references}

\section{Appendix}

\subsection{Overview of Measure Theory}\label{app:meas}
In order to formalize probability distributions with uncountable support such as Gaussians, one traditionally employs measure theory. While measure theory is not central to the understanding of this article, it is needed to make certain arguments and intuitions rigorous, so we provide a short reference and refer to \citep{kallenberg} for a comprehensive introduction.

\paragraph{Measurable spaces}
A \emph{$\sigma$-algebra} on a set $X$ is a collection of subsets of $X$ which contains $\emptyset$ and is closed under complementation and countable union. A \emph{measurable space} is a pair $(X,\Sigma_X)$ of a set $X$ and a $\sigma$-algebra $\Sigma_X$ on $X$. The subsets $U \in \Sigma_X$ are called \emph{measurable subsets}. A function $f : X \to Y$ between measurable spaces is called \emph{measurable} if for all measurable subsets $A \subseteq Y$, the preimage $f^{-1}(A)$ is measurable in $X$. The product space $X \times Y$ is naturally equipped with the product-$\sigma$-algebra, which is generated by the rectangles $A \times B$ for $A \in \Sigma_X, B \in \Sigma_Y$.

Every topological space $X$ comes with a natural $\sigma$-algebra, namely the \emph{Borel $\sigma$-algebra} $\mathcal B(X)$ generated by its open sets. Continuous functions between topological spaces become measurable under this definition.

All examples of measurable spaces in this article will be of a particularly well-behaved type, namely \emph{standard Borel spaces}. 
\begin{definition}
A \emph{Polish space} is a topological space homeomorphic to a complete metric space with a countable dense subset. A \emph{standard Borel space} is a measurable space which is isomorphic to $(X,\mathcal B(X))$ for some Polish space $X$.
\end{definition}
We write $\sbs$ for the category of standard Borel spaces and measurable maps. This category includes all spaces that will be relevant for real-valued probability, for example $\R^n$, the interval $[0,1]$ and countable discrete spaces.

\begin{definition}
Let $X$ be a measurable space. A \emph{measure} on $X$ is a function $\mu : \Sigma_X \to [0,\infty]$ such that 
	\begin{equation} \mu(\emptyset) = 0 \qquad \text{ and } \qquad \mu\left(\sum_{i=1}^\infty A_i \right) = \sum_{i=1}^\infty \mu(A_i) \label{eq:def_measure} \end{equation}
	where $\sum_{i} A_i$ denotes disjoint union.
\end{definition}
A \emph{probability measure} is a measure satisfying $\mu(X)=1$. If $\mu$ is a measure on $X$ and $f : X \to Y$ is measurable, the \emph{pushforward measure} $f_*\mu$ is defined by $f_*\mu(A) = \mu(f^{-1}(A))$. For $x \in X$, the \emph{Dirac measure} $\delta_x$ is the probability measure on $X$ defined, using Iverson backet notation, by $\delta_x(A) = [x \in A]$. The Borel-Lebesgue measure is the unique measure on $(\R,\mathcal B(\R))$ assigning every interval its length, that is $\ell([a,b]) = b-a$ for all $a \leq b$. For two probability measures $\mu,\nu$, the product probability measure $\mu \otimes \nu$ is uniquely defined via $(\mu \otimes \nu)(A \times B) = \mu(A) \cdot \nu(B)$ for all measurable $A,B$. A subset which is assigned measure zero can be seen as negligible; this informs the following terminology:

\begin{definition}If $\mu$ is a measure on $X$ and $\phi(x)$ some measurable property, we say $\phi$ holds \emph{$\mu$-almost everywhere} if $\mu(\{ x \in X : \neg \phi(x)\}) = 0$. We say $\mu$ is \emph{absolutely continuous} with respect to $\nu$, written $\mu \ll \nu$, if for all measurable sets $A$, $\nu(A) = 0$ implies $\mu(A) = 0$.
\end{definition}

\noindent\paragraph{Integration}

For a nonnegative measurable function $f : X \to [0,\infty)$ and a measure $\mu$, its integral is defined as
\[ \int_X f(x) \mu(\d x) \defeq \sup_{\{A_i\}} \sum_{i} \mu(U_i) \cdot \inf_{x \in U_i} f(x) \quad \in \quad [0,\infty] \]
where the supremum ranges over finite measurable partitions of $X$. The integral extends to measurable functions $f : X \to \R$ that are integrable, i.e. satisfy $\int |f(x)| \mu(\d x) < \infty$. Fubini's theorem states that the order of integration can be interchanged; for all probability measures $\mu,\nu$ and integrable $f : X \times Y \to \R$, we have
\begin{equation}
\int \int f(x,y) \mu(\d x)\nu(\d y) = \int \int f(x,y) \nu(\d y) \mu(\d x)
\end{equation}
Density functions are important practical tools for defining measures: If $\mu$ is a measure on $X$ and $f : X \to [0,\infty)$ is measurable, then
\[ \nu(A) = \int_A f(x)\mu(\d x) \defeq \int f(x)[x \in A]\mu(\d x) \]
defines another measure. For example, the standard normal distribution on $\R$ is defined as having the density function $\varphi$ with respect to the Lebesgue measure, where
\[ \varphi(x) = \frac{1}{\sqrt{2\pi}}e^{-x^2/2} \] 

\begin{definition}\label{def:kernel}
A probability kernel $X \kerto Y$ between measurable spaces is a function $f : X \times \Sigma_Y \to [0,1]$ such that $f(-,A)$ is measurable for all $A \in \Sigma_Y$ and $f(x,-)$ is a probability measure for all $x \in X$. 
Kernels $f : X \kerto Y$ and $g : Y \kerto Z$ compose using integration
\begin{equation}
(g \bullet f)(x,A) = \int_Y g(y,A)f(x,\mathrm dy) \label{eq:kernel_composition}
\end{equation}
and any measurable map $f : X \to Y$ induces a Dirac kernel $\delta_f :X \kerto Y$ via $\delta_f(x,A) = [f(x) \in A]$.
\end{definition}
\begin{definition}[{\cite{fritz}}]
The category $\catname{BorelStoch}$ consists of standard Borel spaces and probability kernels between them. Identities are given by $\delta_{\id_X}$ and composition is kernel composition.
\end{definition}

$\catname{BorelStoch}$ has the structure of a Markov category, whose tensor is given by $X \otimes Y = X \times Y$ on objects and by product measures on morphisms. Copying and discarding is given by the Dirac kernels for the canonical maps $\Delta : X \to X \times X$, $! : X \to 1$ in $\catname{Sbs}$.   

Our example Markov categories $\catname{FinStoch}$ and $\catname{Gauss}$ faithfully embed in $\catname{BorelStoch}$, where we interpret
\begin{enumerate}
\item the finite set $X$ as the discrete standard Borel space $(X, \mathcal P(X))$
\item the object $n$ of $\gauss$ as $(\R^n,\mathcal B(\R^n))$
\end{enumerate}
and the morphisms of these categories as actual probability kernels. 

\begin{definition}[Giry monad]
	\label{def:giry_monad}
	There is a monad $\G : \sbs \to \sbs$ due to Giry \cite{giry} that assigns to $X$ the space of probability measures $\G X$, endowed with the least $\sigma$-algebra making all evaluations $\ev_A : \G X \to [0,1], \mu \mapsto \mu(A)$ measurable for $A \in \Sigma_X$. The unit of this monad takes the Dirac measure $x \mapsto \delta_x$. Kleisli composition takes the average measure via integration, that is for $f : X \to \G Y$ and $\mu \in \G X$, we have the Kleisli extension
	\begin{equation}
	f^\dagger(\mu)(A) = \int_X f(x)(A) \mu(\mathrm dx) \label{eq:giry_bind}
	\end{equation}
	For $h : X \to Y$, the functorial action $\G(h)(\mu)$ is given by the pushforward measure $h_*\mu$.
\end{definition}
The Giry monad is strong, affine and commutative, where commutativity follows from Fubini's theorem. A Kleisli arrow $X \to \G Y$ is the same as a Markov kernel $X \kerto Y$, and Kleisli composition \eqref{eq:giry_bind} agrees with kernel composition \eqref{eq:kernel_composition}.

\subsection{CD-calculus}\label{app:cdcalc}
We verify that the remaining axioms of the CD calculus hold in any CD model, completing the proof from Section~\ref{cd:equations}. \vspace{-1cm}
\begin{proof}[Proof of \Cref{prop:cd_soundness}] 
\eqref{cd:let_cong} follows from the compositionality of the semantics. \eqref{cd:pair_beta}, \eqref{cd:pair_eta} and \eqref{cd:unit_eta} are are immediate. We proceed to prove that \eqref{cd:let_lin} is valid, that is if $t$ uses $x$ exactly once, then 
\[ \input{categorical_probability/eq_let_lin.tikz} \]
We argue by induction over the term structure of $t$. 

\paragraph{Variable} If $t=x$, then 
\[ \input{categorical_probability/eq_let_lin_x.tikz} \]
\paragraph{Pairing} Let $t=(u,s)$ where wlog $x$ occurs freely exactly once in $u$ and zero times in $s$. 
By inductive hypothesis, we have $(\letin x e u) = u[e!x]$, i.e.
\[ \input{categorical_probability/eq_let_lin_u.tikz} \]
from which we derive by the comonoid laws and weakening of $s$
\[ \input{categorical_probability/eq_let_lin_pair.tikz} \]
The case for $(s,u)$ is symmetric.
\paragraph{Function application} If $t=f\,u$ we obtain immediately from the inductive hypothesis  
\[ \input{categorical_probability/eq_let_lin_f.tikz} \]
The proof for the projection case $t=\pi_i\,u$ is analogous.
\paragraph{Let-binding I} Let $t=(\letin y u s)$ with $u,s$ as before then
\[ (\letin x e \letin y u s) \equiv (\letin y {(\letin x e u)} s) \equiv (\letin y {u[e!x]} s) \]
is a special case of \eqref{cd:assoc} which was proved in \eqref{cd:eq_assoc}.
\paragraph{Let-binding II} Let $t = (\letin y s u)$, then
\[ (\letin x e \letin y s u) \equiv (\letin y s \letin x e u) \equiv (\letin y s u[e!x]) \]
is a special case of \eqref{cd:comm}. The inductive hypothesis on $u$ involves both weakening and exchange and reads
\[ \input{categorical_probability/eq_let_lin_let_hyp.tikz} \]
from which we derive
\[ \input{categorical_probability/eq_let_lin_let.tikz} \]
This finishes the validation for linear substitution. \\

For \eqref{cd:let_val}, we carry out the semantic analogue of \Cref{prop:cd_subst} and consider sequences of let-bindings
\[ \letin {x_1} e \cdots \letin {x_n} e \hat t \]
whose denotation is
\[ \input{categorical_probability/eq_let_val.tikz} \]
This can be replaced by $\letin x e \letin {x_1} x \cdots \letin {x_n} x \hat t$, i.e.
\[ \input{categorical_probability/eq_let_val_2.tikz} \]
whenever the denotation of $e$ is deterministic in the CD category sense. It remains to note that the denotations of all values of the CD calculus are always deterministic.
\end{proof}

\begin{proof}
	The first case is \eqref{cd:let_lin} and the last case follows from the combination of the previous ones. 
	
	\paragraph{Copyable} We begin with the special case that $t$ has precisely two occurrences of $x$. Then
	\begin{align*}
	&t[e/x] \\
	\stackrel{\eqref{cd:general_subs}}\equiv\, &\letin {x_1} e \letin {x_2} e \hat t \\
	\stackrel{\eqref{cd:let_val},\eqref{cd:pair_beta}}\equiv\, &\letin p {(e,e)} \letin {x_1} {\pi_1\,p} \letin {x_2} {\pi_2 p} \hat t \\
	\stackrel{\eqref{cd:copyable}}\equiv\, &\letin p {(\letin x e (x,x))} \letin {x_1} {\pi_1\,p} \letin {x_2} {\pi_2 p} \hat t \\
	\stackrel{\eqref{cd:assoc}} \equiv\, & \letin x e \letin p {(x,x)} \letin {x_1} {\pi_1\,p} \letin {x_2} {\pi_2 p} \hat t \\
	\stackrel{\eqref{cd:let_val},\eqref{cd:pair_beta}} \equiv\, & \letin x e \letin {x_1} {x} \letin {x_2} {x} \hat t \\
	\stackrel{\eqref{cd:general_subs}} \equiv\, & \letin x e t
	\end{align*}
	Repeating this process, any chain of repeated let bindings of a copyable term $e$
	\[ \letin {x_1} e \cdots \letin {x_n} e \ldots \]
	can be replaced by 
	\[  \letin x e \letin {x_1} x \cdots \letin {x_n} x \ldots \]
	\paragraph{Discardable} Let $t$ have no free occurrence of $x$, and $e$ be discardable. Then
	\begin{align*}
	&\letin x e t \\
	\stackrel{\eqref{cd:let_val}} \equiv\, &\letin x e \letin y {()} t \\
	\stackrel{\eqref{cd:assoc}} \equiv\, &\letin y {(\letin x e ())} t \\
	\stackrel{\eqref{cd:discardable}} \equiv\, &\letin y {()} t \\
	\stackrel{\eqref{cd:let_val}} \equiv\, &t
	\end{align*}
\end{proof}

\begin{proof}[Proof of \Cref{prop:all_lambdac_axioms}]
	We will invoke \eqref{cd:let_cong} implicitly throughout. \eqref{cd:let_beta} follows immediately from \eqref{cd:let_val} because $x_2$ is a value. \eqref{cd:id}, \eqref{cd:let_f}, \eqref{cd:let_ast} follow by applying \eqref{cd:let_lin} one or two times.
	
	For \eqref{cd:comm}, we notice that because $x_2 \notin \fv(e_1)$, the expression $\letin {x_1} {e_2} {\letin {x_2} {{\color{red} x_2}} e}$ has a unique free occurrence of $x_2$, hence by linear substitution
	\begin{align*}
		&\letin {x_2} {e_2} \letin {x_1} {e_1} {\color{blue} e} \\
		\overset{\eqref{cd:let_beta}}\equiv\;& \letin {{\color{red} x_2}} {{e_2}} \letin {x_1} {e_1} {\color{blue} \letin {x_2} {{\color{red} x_2}} e} \\
		\overset{\eqref{cd:let_lin}}\equiv\;&  \letin {x_1} {e_1} \letin {x_2} {{\color{red} e_2}} e
	\end{align*}
	For \eqref{cd:assoc}, if $x_1 \notin \fv(e)$ then $\letin {x_2} {{(\letin {x_1} {{\color{red} x_1}} e_2)}} e$ has a unique free occurrence of $x_1$, hence
	\begin{align*}
		&\letin {x_1} {e_1} \letin {x_2} {{\color{blue} e_2}} e \\
		\overset{\eqref{cd:let_beta}}\equiv\;& \letin {{\color{red} x_1}} {{e_1}} \letin {x_2} {{(\color{blue} \letin {x_1} {{\color{red} x_1}} e_2)}} e \\
		\overset{\eqref{cd:let_lin}}\equiv\;&  \letin {x_2} {(\letin {x_1} {{\color{red} e_1}} e_2)} e \qedhere
	\end{align*}
\end{proof}

\end{document}